\documentclass
{siamart171218}

\usepackage{amsmath}
\usepackage{amssymb}
\usepackage{amsfonts}
\usepackage{graphicx}
\usepackage{color}
\usepackage{datetime}

\usepackage{lipsum}
\usepackage{amsfonts}
\usepackage{graphicx}
\usepackage{algorithmic}

\newsiamremark{remark}{Remark}
\newsiamremark{hypothesis}{Hypothesis}
\crefname{hypothesis}{Hypothesis}{Hypotheses}
\newsiamthm{claim}{Claim}

\headers{Moment Dynamics and Filtering for Quasilinear 
Quantum Systems}{I.G.Vladimirov, and I.R.Petersen}

\def\<{\leqslant}           
\def\>{\geqslant}           

\def\d{\partial}

\def\wt{\widetilde}

\def\Re{\mathrm{Re}}   
\def\Im{\mathrm{Im}}   

\def\mR{\mathbb{R}}    
\def\mC{\mathbb{C}}    

\def\Tr{\mathrm{Tr}}       
\def\rT{\mathrm{T}}        
\def\rF{\mathrm{F}}        
\def\diam{\diamond}       

\def\bE{\mathbf{E}}    

\def\[[[{[\![\![}
\def\]]]{]\!]\!]}

\def\bra{\langle}
\def\ket{\rangle}

\def\re{\mathrm{e}}        
\def\rd{\mathrm{d}}        

\def\x{\times}
\def\ox{\otimes}

\def\fF{\mathfrak{F}}

\def\fH{\mathfrak{H}}

\def\cZ{\mathcal{Z}}
\def\fZ{\mathfrak{Z}}

\def\bzero{\mathbf{0}}

\def\cX{\mathcal{X}}

\def\cG{\mathcal{G}}
\def\cI{\mathcal{I}}

\def\cov{\mathbf{cov}}

\def\eps{\epsilon}
\def\Ups{\Upsilon}
\def\ups{\upsilon}

\def\lprod{\mathop{\overleftarrow{\prod}}}
\def\lexp{\mathop{\overleftarrow{\exp}}}

\DeclareMathAlphabet{\mathbfit}{OML}{cmm}{b}{it}

\title{Moment Dynamics 
and 
Observer Design 
for A Class of Quasilinear Quantum Stochastic Systems\thanks{This work is supported by the Australian Research Council under grant DP180101805 and the Office of Naval Research under grant N62909-19-1-2129.
}}

\author{Igor G. Vladimirov$^{\dagger}$, \quad Ian R. Petersen
\thanks{Research School of Electrical, Energy and Materials Engineering, College of Engineering and Computer Science, Australian National University, Canberra, ACT 2601, Australia,
  \email{igor.g.vladimirov@gmail.com, i.r.petersen@gmail.com}.
  }
  }

\usepackage{amsopn}

\begin{document}

\maketitle

\begin{abstract}
This paper is concerned with a class of open quantum systems whose dynamic  variables have an algebraic structure, similar to that of the Pauli matrices pertaining to finite-level systems. The system interacts with external bosonic fields,  and its Hamiltonian and coupling operators depend linearly on the system variables. This results in a Hudson-Parthasarathy quantum stochastic differential equation (QSDE) whose drift and dispersion terms are affine and linear functions of the system variables. The quasilinearity of the QSDE leads to tractable dynamics of mean values and higher-order multi-point moments of the system variables driven by vacuum input fields. This allows  for the closed-form computation of  the quasi-characteristic function  of the invariant quantum state of the system and infinite-horizon asymptotic growth rates for a class of cost functionals. The tractability of the moment dynamics is also used for mean square optimal Luenberger observer design in a measurement-based filtering problem for a quasilinear quantum plant, which leads to a Kalman-like quantum filter.
\end{abstract}



\section{Introduction}

The present paper is concerned with a stochastic control theoretic approach to a class of open quantum systems,  
such as the electron spin interacting with a quantised  electromagnetic field (for example, nonclassical light from a laser). The analysis and synthesis of such systems with desired properties are  important for quantum communication, information and computing technologies \cite{NC_2000} and other modern engineering developments which employ quantum mechanical resources of light-matter interaction. The Hudson-Parthasarathy calculus \cite{HP_1984,P_1992} provides a paradigm for modelling such systems in the form of quantum stochastic differential equations (QSDEs). These equations describe the Heisenberg evolution of the quantum system variables and are   driven by quantum Wiener processes which represent  the external bosonic fields. A particular form of the dynamics of the system and field variables (as time-varying operators on a Hilbert space) depends on their algebraic properties and also on the structure of the Hamiltonian and coupling operators (as functions of the system variables) which specify the energetics of the system and its interaction with the surroundings.

In this work, we consider a class of quantum stochastic systems whose variables have an algebraic structure,    similar to that of the Pauli matrices \cite{S_1994} which play an important role in the theory of quantum angular momentum and other quantum mechanical models (such as finite-level systems). The system Hamiltonian and the system-field coupling operators are assumed to be linear and affine functions of the system variables, to which more complicated  nonlinearities in the energy operators reduce due to the algebraic structure being considered.  This setting leads to QSDEs whose drift vector depends affinely and the dispersion matrix depends linearly on the system variables, which is similar to classical SDEs with a multiplicative noise \cite{W_1967}. Such QSDEs differ from linear QSDEs for open quantum harmonic oscillators \cite{NY_2017,P_2017} and are ``quasilinear'' in the sense that their right-hand side   involves a bilinear dependence on the system and field variables. Nevertheless, the  solutions of these quasilinear  QSDEs admit a representation in terms of time-ordered operator exponentials which play the role of fundamental solutions. Despite being more complicated than the standard matrix exponentials, these fundamental solutions  lend themselves to effective quantum averaging in the case when the external fields are in the vacuum quantum state \cite{P_1992}. In combination with the underlying algebraic structure, this makes mean values, covariances and a wide class of one-point nonlinear moments of the system variables amenable to practical computation along with their multi-point mixed moments.
Under a stability condition,  the tractability of the dynamics of moments allows their limit values to be computed in closed form along with the quasi-characteristic function \cite{CH_1971}  of the invariant quantum state of the system and infinite-horizon asymptotic growth rates for a class of integral cost functionals. The first and second-order moment dynamics are used in a quantum filtering \cite{B_1983,B_1989,BVJ_2007}  problem, where  the output of the quantum system is converted by a measuring device to a  self-commuting multivariable  observation process, which drives a classical SDE for a linear Luenberger observer. This observer produces an unbiased estimate of the system variables,  and its gain matrix lends itself to optimisation by a minimum mean square estimation error criterion. The resulting mean square optimal Luenberger observer has the structure of the Kalman filter \cite{AM_1979,LS_2001} and involves a Riccati equation for the estimation error covariance matrix.

Similar quasilinear QSDEs were obtained in \cite{VP_2012c} for a class of quantum anharmonic oscillators whose Hamiltonians and coupling operators are, respectively, cubic and quadratic functions of quantum position and momentum variables. However, the quantum systems,  discussed here, are different and correspond to those in the works \cite{EMPUJ_2012,EMPUJ_2016}, which consider a QSDE with an affine drift and a linear dispersion matrix as a starting point and studies conditions on its coefficients to guarantee that the QSDE  preserves the 
commutation and anticommutation relations, necessary for physical realizability of such a QSDE. However, the present paper starts from particular Hamiltonian and coupling operators,  leading  to QSDEs with affine drifts and linear dispersion matrices which are physically realizable  and preserve the algebraic structure of the system variables by construction. Furthermore, the main focus of this paper is on quantum statistical aspects of such systems and their application to quantum filtering, with a view of extending the quantum adaptation of the method of moments to mean square optimal and more complicated nonquadratic
control problems \cite{DP_2010,WM_2010} for this class of open quantum systems.

The paper is organised as follows.
Section~\ref{sec:var} specifies quantum dynamic variables with an algebraic structure and discusses their boundedness and commutation properties. Section~\ref{sec:QSS} 
obtains a quasilinear QSDE  for the quantum system with a linear Hamiltonian and affine coupling operators and describes its solutions in terms of time-ordered exponentials. 
Section~\ref{sec:mom} discusses the dynamics of mean values and more general nonlinear and multipoint mixed moments of the system variables in the case of vacuum input fields along with the invariant quantum state and cost functional growth rates. 
Section~\ref{sec:filt} applies these results to a mean square optimal filtering problem for the quantum plant using measurement-based  Luenberger observers, including a steady-state regime.
Section~\ref{sec:Pauli} provides an example of stability conditions for quantum plants with the Pauli matrices as initial variables.
Section~\ref{sec:conc} makes concluding remarks.

\section{Quantum variables with an algebraic structure}
\label{sec:var}

We consider a quantum system with $n$ dynamic variables $X_1, \ldots, X_n$ which are time-varying
self-adjoint operators on an underlying  complex separable Hilbert space $\fH$ (their time evolution will be specified in Section~\ref{sec:QSS}). These quantum variables (taken at the same moment of time) are assumed to have an algebraic structure
\begin{equation}
\label{XXX}
    \Xi_{jk}
    :=
    X_j X_k  = \alpha_{jk}\cI_{\fH} + \sum_{\ell=1}^n\beta_{jk\ell} X_\ell,
    \qquad
    j,k=1, \ldots, n,
\end{equation}
where $\alpha:= (\alpha_{jk})_{1\< j,k\< n} \in \mC^{n\x n}$ is a complex matrix and $\beta:= (\beta_{jk\ell})_{1\< j,k,\ell\< n} \in \mC^{n\x n\x n}$ is a complex array, which consist of structure constants, and $\cI_{\fH}$ is the identity operator on $\fH$. Denoted by $\beta_\ell:= (\beta_{jk\ell})_{1\< j,k\< n} \in \mC^{n\x n}$ are ``sections'' of $\beta$, with $\ell =1, \ldots, n$. A vector-matrix form of (\ref{XXX}) is
\begin{equation}
\label{XXX1}
    \Xi:=
    (\Xi_{jk})_{1\< j,k\< n}
    =
    XX^\rT =  \alpha\ox \cI_{\fH} + \beta \cdot X,
\end{equation}
where
\begin{equation}
\label{betaX}
    \beta \cdot X
    :=
    \sum_{\ell = 1}^n
    \beta_\ell  X_\ell
\end{equation}
is an $(n\x n)$-matrix of operators, and the system variables are assembled into the vector
$
    X:= (X_k)_{1\< k \< n}
$. Here, vectors are organised as columns unless indicated otherwise, and the transpose $(\cdot)^{\rT}$ acts on vectors and matrices of operators as if their entries were scalars. Also,
$\ox$ is the tensor product of operators (in particular, the Kronecker product of matrices) or spaces. The matrices $\alpha$ and $\alpha\ox \cI_{\fH}$ will be identified with each other. An alternative representation of the matrix $\beta\cdot X$ in (\ref{betaX}) is
$$
    \beta\cdot X =
    \begin{bmatrix}
      \beta_{\bullet 1 \bullet} X
      &
      \ldots &
      \beta_{\bullet n \bullet} X
    \end{bmatrix},
$$
so that its $k$th column $\beta_{\bullet k \bullet} X = \big(\sum_{\ell = 1}^n\beta_{jk\ell} X_\ell\big)_{1\< j\< n}$ involves another section $\beta_{\bullet k \bullet} := (\beta_{jk\ell})_{1\< j,\ell \< n} \in \mC^{n\x n}$ of the array $\beta$. We will also use a different product of $\beta$ (or a  similar array) with a vector $u \in \mC^n$:
\begin{equation}
\label{diam}
    \beta \diam u
    :=
    \begin{bmatrix}
        \beta_1 u & \ldots & \beta_n u
    \end{bmatrix}
    \in \mC^{n\x n}.
\end{equation}
The products in (\ref{betaX}), (\ref{diam}) are related as
\begin{equation}
\label{cdotdiam}
  (\beta \cdot u)v = \sum_{\ell=1}^n \beta_\ell u_\ell  v
  =
  \sum_{\ell=1}^n \beta_\ell v  u_\ell  = (\beta\diam v)u
\end{equation}
for any   $
  u:=(u_k)_{1\< k\< n},v:=(v_k)_{1\< k\< n} \in \mC^n$, which holds for any $\mC^{n\x n\x n}$-valued array instead of $\beta$ and extends to vectors $u$, $v$ of $n$ quantum variables such that $[u,v^\rT] = 0$.
As
operators  on $\fH$, the entries $\Xi_{jk}$ of the matrix $\Xi$ in (\ref{XXX1}), defined by the first equality in (\ref{XXX}),   satisfy
\begin{equation}
\label{Xi+}
  \Xi_{jk}^\dagger = X_k^\dagger X_j^\dagger = \Xi_{kj}
\end{equation}
in view of the self-adjointness of $X_1, \ldots, X_n$. Here, $(\cdot)^\dagger$ is the operator adjoint  which extends to matrices of operators as the transpose $(\cdot)^\dagger: = ((\cdot)^\#)^\rT$ of the entrywise adjoint $(\cdot)^\#$.   In matrix form, the identities (\ref{Xi+}) are represented as $\Xi^\dagger = \Xi$. In order to guarantee the consistency
\begin{equation}
\label{XX+}
    \overline{\alpha_{jk}} + \sum_{\ell=1}^n\overline{\beta_{jk\ell}} X_\ell
    =
    \alpha_{kj} + \sum_{\ell=1}^n\beta_{kj\ell} X_\ell
\end{equation}
of the relations (\ref{XXX}), (\ref{Xi+})
for all $j,k=1,\ldots, n$, the matrices $\alpha$ and $\beta_1, \ldots, \beta_n$ are assumed to be Hermitian:
\begin{equation}
\label{herm}
  \alpha^* = \alpha,
  \qquad
  \beta_\ell^* = \beta_\ell,
  \qquad
  \ell = 1, \ldots, n
\end{equation}
(with $(\cdot)^*:= {\overline{(\cdot)}}^\rT$ the complex conjugate transpose),
which is equivalent to the symmetry of
the matrices $\Re \alpha$, $\Re \beta_1, \ldots, \Re \beta_n$ and antisymmetry of $\Im \alpha$, $\Im \beta_1, \ldots, \Im \beta_n$.  The equalities (\ref{XX+}) are obtained from (\ref{Xi+})
  by taking the adjoint of the right-hand sides of (\ref{XXX}) and using the self-adjointness of $X_1, \ldots, X_n$.  In fact, (\ref{herm}) is not only sufficient for (\ref{XX+}) but is also necessary if
\begin{equation}
\label{indep}
  {\rm the\ operators}\ \cI_\fH, X_1, \ldots, X_n\ {\rm are\ linearly\ independent}.
\end{equation}
We will now discuss several corollaries of (\ref{XXX}), which are used in what follows.
The algebraic property (\ref{XXX}) allows a quadratic function of the system variables (considered at the same moment of time)   to be reduced to an affine function. More precisely (cf. a similar remark in \cite[paragraph~4 on p.~641]{EMPUJ_2016}),
\begin{align}
\nonumber
    X^\rT R X
    & = \sum_{j,k=1}^n r_{jk}\Xi_{jk}
    = \sum_{j,k=1}^n r_{jk} \Big(\alpha_{jk} + \sum_{\ell=1}^n\beta_{jk\ell} X_\ell\Big)\\
\nonumber
    & =
    \bra R, \alpha\ket_\rF +
    \begin{bmatrix}
      \bra R, \beta_1\ket_\rF &
      \ldots &
      \bra R, \beta_n\ket_\rF
    \end{bmatrix}
    X\\
\label{XRX}
    & =
    \begin{bmatrix}
        \bra R, \alpha\ket_\rF &
      \bra R, \beta_1\ket_\rF &
      \ldots &
      \bra R, \beta_n\ket_\rF
    \end{bmatrix}
    \begin{bmatrix}
      \cI_\fH\\
      X
    \end{bmatrix}
\end{align}
for any real symmetric matrix $R:= (r_{jk})_{1\< j,k\< n} = R^\rT \in \mR^{n\x n}$, where $\bra K, N\ket_\rF:= \Tr(K^* N)$ is the Frobenius  inner product  of real or complex matrices with the Frobenius norm $\|K\|_\rF:= \sqrt{\bra K, K\ket_\rF}$. A similar reduction to an affine function  holds for an arbitrary polynomial of the system variables. This reduction yields an unambiguous result if (\ref{XXX}) is consistent with the associative property $(\xi\eta)\zeta = \xi (\eta\zeta)$ of the algebra of linear operators on $\fH$,  as discussed below, where an important role is played by the condition (\ref{indep}).

\begin{theorem}
\label{th:ass}
The following equalities are sufficient and, under the condition (\ref{indep}), necessary for the
relations (\ref{XXX}) (or (\ref{XXX1})) to be consistent with the associativity of the operator multiplication:
\begin{align}
\label{con1}
    \sum_{\ell = 1}^n
    (\alpha_{\ell s}
    \beta_{jk\ell}
    -
    \alpha_{j\ell }
    \beta_{ks \ell }) & = 0,\\
\label{con2}
        \alpha_{jk} \delta_{rs}
        -
        \alpha_{ks} \delta_{rj}
        +
        \sum_{\ell = 1}^n
        (\beta_{jk\ell} \beta_{\ell s r}
        -
        \beta_{ks\ell} \beta_{j\ell r}) & = 0,
        \qquad
        j,k,s,r=1, \ldots, n,
\end{align}
where $\delta_{jk}$ is the Kronecker delta. \hfill$\square$
\end{theorem}
\begin{proof}
Repeated application of (\ref{XXX}) to computing the product $X_jX_kX_s$ of the system variables in two ways, as $(X_jX_k)X_s$ and $X_j(X_kX_s)$,  leads to
\begin{align}
\nonumber
    (X_jX_k)X_s
    & =
    \Big(
        \alpha_{jk}
        +
        \sum_{\ell = 1}^n
        \beta_{jk\ell}X_\ell
    \Big)
    X_s
    =
    \alpha_{jk}X_s
    +
    \sum_{\ell = 1}^n
    \beta_{jk\ell}
    \Big(
        \alpha_{\ell s}
        +
        \sum_{r=1}^n
        \beta_{\ell s r}
        X_r
    \Big)\\
\label{Xjks1}
    & =
    \sum_{\ell = 1}^n
    \alpha_{\ell s}
    \beta_{jk\ell}
    +
    \sum_{r=1}^n
    \Big(
        \alpha_{jk} \delta_{rs}
        +
        \sum_{\ell = 1}^n
        \beta_{jk\ell} \beta_{\ell s r}
    \Big)
    X_r,\\
\nonumber
    X_j(X_kX_s)
    & =
    X_j
    \Big(
        \alpha_{ks}
        +
        \sum_{\ell = 1}^n
        \beta_{ks\ell}X_\ell
    \Big)
    =
    \alpha_{ks}X_j
    +
    \sum_{\ell = 1}^n
    \beta_{ks\ell}
    \Big(
        \alpha_{j\ell}
        +
        \sum_{r=1}^n
        \beta_{j\ell r}
        X_r
    \Big)\\
\label{Xjks2}
    & =
    \sum_{\ell = 1}^n
    \alpha_{j\ell }
    \beta_{ks \ell }
    +
    \sum_{r=1}^n
    \Big(
        \alpha_{ks} \delta_{rj}
        +
        \sum_{\ell = 1}^n
        \beta_{ks\ell} \beta_{j\ell r}
    \Big)
    X_r .
\end{align}
By comparing the coefficients before $\cI_\fH$, $X_1, \ldots, X_n$ on the right-hand sides of (\ref{Xjks1}), (\ref{Xjks2}), it follows that the fulfillment of (\ref{con1}), (\ref{con2}) leads to $(X_jX_k)X_s = X_j(X_kX_s)$ for all $j,k,s=1, \ldots, n$ in accordance with the associative property of the operator multiplication, thus proving the sufficiency. Under the condition (\ref{indep}),   the necessity of  (\ref{con1}), (\ref{con2}) for the consistency of (\ref{XXX}) with the associativity follows from the implication $c_0 + \sum_{r=1}^n c_r X_r = 0 \Longrightarrow c_0 = \ldots = c_n = 0$ for any $c_0, \ldots, c_n \in \mC$ applied to $(X_jX_k)X_s-X_j(X_kX_s) = 0$ as the difference between  the right-hand sides of (\ref{Xjks1}), (\ref{Xjks2}).
\end{proof}

The above proof employs reduction of degree three monomials in the system variables to affine functions. The conditions (\ref{con1}), (\ref{con2}), which are organised as quadratic constraints on the structure constants, are assumed to be satisfied in what follows.

As discussed in the following theorem, the algebraic property (\ref{XXX}) also implies that the quantum variables $X_1, \ldots, X_n$ are  bounded operators on the Hilbert space $\fH$.

\begin{theorem}
\label{th:bound}
Under the conditions (\ref{XXX1}), (\ref{herm}), the induced norms of the system variables satisfy
\begin{equation}
\label{Xnorm}
  \|X_k\| \< \frac{1}{2} |\tau_k| + \gamma,
  \qquad
  k = 1, \ldots, n,
\end{equation}
where
\begin{equation}
\label{tau}
      \tau
      :=
      (\tau_\ell)_{1\< \ell \< n}
      :=
      (\Tr \beta_\ell)_{1\< \ell \< n} \in \mR^n
\end{equation}
is a vector formed from the traces of the Hermitian sections $\beta_1, \ldots, \beta_n$ of the array $\beta$, and
\begin{equation}
\label{rad}
  \gamma := \sqrt{\Tr \alpha + \frac{1}{4}|\tau|^2}.
\end{equation}
\hfill$\square$
\end{theorem}
\begin{proof}
Consider the quantum covariance matrix of the system variables, which is a complex positive semi-definite Hermitian matrix
\begin{equation}
\label{covX}
    \cov(X):= \bE \Xi-\mu\mu^\rT = \alpha + \beta\cdot \mu - \mu\mu^\rT,
\end{equation}
obtained by averaging (\ref{XXX1}). Here,
\begin{equation}
\label{mu}
      \mu:= (\mu_k)_{1\< k\< n}:= \bE X \in \mR^n
\end{equation}
is the mean vector, and
\begin{equation}
\label{bE}
    \bE \zeta := \Tr(\rho\zeta)
\end{equation}
is the quantum expectation \cite{H_2001} over a density operator $\rho$ which is a positive semi-definite self-adjoint operator  on $\fH$ of unit trace $\Tr \rho =1$. Since the matrix  $\cov(X)\succcurlyeq 0$ has a nonnegative trace, then (\ref{tau}), (\ref{covX}) imply that
\begin{equation}
\label{pos2}
    \Tr \alpha + \tau^\rT \mu  - |\mu|^2 \> 0 ,
\end{equation}
where $\Tr \alpha$, $\tau_1, \ldots, \tau_n$ are real-valued   as the traces of Hermitian matrices in (\ref{herm}). By completion of the square, (\ref{pos2}) is equivalent to
\begin{equation}
\label{pos}
    \big|\mu-\frac{1}{2}\tau\big|^2 \< \Tr \alpha + \frac{1}{4}|\tau|^2,
\end{equation}
whereby the right-hand side is nonnegative, thus giving rise to (\ref{rad}).
By the triangle inequality, (\ref{pos}) implies that
\begin{equation}
\label{muabs}
    |\mu_k|
    \<
    \frac{1}{2}|\tau_k|
    +
    \big|\mu_k-\frac{1}{2}\tau_k\big|
    \<
    \frac{1}{2}|\tau_k| + \gamma,
    \qquad
    1\< k \< n.
\end{equation}
Now, $\|\zeta\| = \sup_{\rho} |\Tr (\rho\zeta)|$ for any self-adjoint operator $\zeta$, where the supremum is over all possible density operators $\rho$ on $\fH$ and can be reduced to pure states $\rho = |\psi\ket \bra \psi|$, represented using the quantum mechanical bra-ket notation \cite{S_1994}, where $\psi$ is an element of $\fH$ of unit norm (so that,  $\Tr \rho = \bra \psi\mid\psi\ket = 1$). Therefore, since the right-hand side of (\ref{muabs}) does not depend on $\rho$, the maximization of its left-hand side as $\|X_k\| = \sup_{\rho}|\mu_k|$ establishes (\ref{Xnorm}).
\end{proof}

The quantity $\gamma$ in (\ref{rad}), which specifies the radius $\gamma$ of a ball of centre $\frac{1}{2}\tau$ containing the mean vector $\mu$ from (\ref{mu}) in view of (\ref{pos}), 
is invariant under translations of the system variables. More precisely, consider self-adjoint quantum variables $\wt{X}_1, \ldots, \wt{X}_n$ obtained by using a ``shift'' vector $\varphi:=(\varphi_k)_{1\< k \< n} \in \mR^n$ as
$$
    \wt{X}:= (\wt{X}_k)_{1\< k \< n} := X + \varphi.
$$
These new variables inherit an algebraic structure from the original system variables $X_1, \ldots, X_n$:
\begin{align*}
    \wt{X}\wt{X}^\rT
    & =
    \Xi + X\varphi^\rT + \varphi X^\rT + \varphi\varphi^\rT
    =
    \alpha + \beta\cdot X + X\varphi^\rT + \varphi X^\rT + \varphi\varphi^\rT\\
    & =
    \alpha - \beta\cdot \varphi -\varphi\varphi^\rT
    + \beta\cdot \wt{X} + \wt{X}\varphi^\rT + \varphi \wt{X}^\rT
=
    \wt{\alpha} + \wt{\beta}\cdot \wt{X},
\end{align*}
which is similar to (\ref{XXX1}). Here, the appropriately modified structure  constants in $\wt{\alpha}$ and $\wt{\beta}:= (\wt{\beta}_{jk\ell})_{1\< j,k,\ell\< n}$  are given by
\begin{equation}
\label{newab}
    \wt{\alpha}= \alpha - \beta\cdot \varphi - \varphi\varphi^\rT,
    \qquad
    \wt{\beta}_{jk\ell} = \beta_{jk\ell} + \varphi_k \delta_{j\ell}+ \varphi_j \delta_{k\ell},
    \qquad
    1\<  j,k,\ell\< n.
\end{equation}
Hence, the entries of the vector $\tau$ in (\ref{tau}) are modified as  $\wt{\tau}_\ell = \sum_{j=1}^n \wt{\beta}_{jj\ell} = \tau_\ell + 2\varphi_\ell$, so that
\begin{equation}
\label{taut}
      \wt{\tau}
      :=
      (\wt{\tau}_\ell)_{1\< \ell \< n}
      =
      \tau + 2\varphi.
\end{equation}
From (\ref{newab}), (\ref{taut}), it now follows that the quantity (\ref{rad}) indeed remains unchanged:
\begin{align*}
    \Tr \wt{\alpha} + \frac{1}{4}|\wt{\tau}|^2
    & =
    \Tr (\alpha - \beta\cdot \varphi - \varphi\varphi^\rT) + \frac{1}{4}|\tau + 2\varphi|^2\\
    & =
    \Tr \alpha
    -\tau^\rT \varphi - |\varphi|^2
    + \frac{1}{4}|\tau|^2+\tau^\rT \varphi + |\varphi|^2
    =
    \Tr \alpha + \frac{1}{4}|\tau|^2.
\end{align*}
Another corollary of the algebraic properties (\ref{XXX}), (\ref{herm}) is provided by the canonical commutation relations
(CCRs)
  \begin{equation}
  \label{XCCR}
    [X_j,X_k]
     =
     \Xi_{jk} - \Xi_{kj}
      = 2i \Im \Xi_{jk}
      =
    2i \Big(\Im \alpha_{jk}  + \sum_{\ell=1}^n(\Im \beta_{jk\ell}) X_\ell\Big),
\end{equation}
where $[a,b]:= ab-ba$
is the commutator of linear operators, use is made of (\ref{Xi+}), and the imaginary part $\Im(\cdot)$ is extended from complex numbers to quantum variables as $\Im \zeta : = \frac{1}{2i}(\zeta - \zeta^\dagger)$. Similarly to (\ref{XXX1}), the CCRs (\ref{XCCR}) are represented in vector-matrix form as
  \begin{equation}
  \label{XCCR1}
    [X,X^\rT]
     :=
    ([X_j,X_k])_{1\< j,k\< n}
    =
    2i (\Im \alpha + (\Im\beta)\cdot X).
  \end{equation}
The operator algebra associativity conditions (\ref{con1}), (\ref{con2}) on $\alpha$, $\beta$ imply the Jacobi identities \cite{D_2006}
\begin{equation*}
\label{Jac1}
    [[X_j,X_k],X_\ell] + [[X_k,X_\ell],X_j] + [[X_\ell,X_j],X_k]  = 0,
    \qquad
    j,k,\ell = 1, \ldots, n.
\end{equation*}
For simplicity, in what follows, $\alpha$ is assumed to be a real symmetric matrix of order $n$,
\begin{equation}
\label{Imalpha0}
    \Im \alpha = 0,
\end{equation}
and hence, the CCRs (\ref{XCCR1}) reduce to
\begin{equation}
\label{XCCRTheta}
    [X,X^\rT]
    =
    2i \Theta \cdot X
    =
    2i
    \sum_{\ell=1}^n
    \Theta_\ell X_\ell.
\end{equation}
Here,
\begin{equation}
\label{Theta}
    \Theta : = (\theta_{jk\ell})_{1\< j,k,\ell \< n}:=  \Im \beta
\end{equation}
is a real $(n\x n\x n)$-array whose sections $\Theta_\ell:= (\theta_{jk\ell})_{1\< j,k\< n} = \Im \beta_\ell\in \mR^{n\x n}$ are antisymmetric matrices for all $\ell = 1, \ldots, n$ in view of (\ref{herm}).

An example of $n=3$ quantum variables, satisfying (\ref{XXX}), with (\ref{herm}),  (\ref{Imalpha0}), is provided by the Pauli matrices \cite{S_1994}
\begin{equation}
\label{X123}
    \sigma_1:=
    \begin{bmatrix}
      0 & 1\\
      1 & 0
    \end{bmatrix},
    \qquad
    \sigma_2:=
    \begin{bmatrix}
      0 & -i\\
      i & 0
    \end{bmatrix},
    \qquad
    \sigma_3:=
    \begin{bmatrix}
      1 & 0\\
      0 & -1
    \end{bmatrix},
\end{equation}
which, together with the identity matrix $I_2$ of order $2$, form a basis in the four-dimensional real space of self-adjoint operators (complex Hermitian $(2\x 2)$-matrices) on the Hilbert space $\fH:= \mC^2$.
In this case, $\alpha = I_3$ and $\beta = i \Theta$, where (\ref{Theta}) is an array $\Theta \in\{0, \pm1\}^{3\x 3\x 3}$ specified by the Levi-Civita symbol $\theta_{jk\ell} = \eps_{jk\ell}$, so that
$\Re \beta = 0$ and
\begin{equation}
\label{T123}
    \Theta_1 =
    \begin{bmatrix}
     0&      0  &   0\\
     0&     0   &  1\\
     0&    -1   &  0
    \end{bmatrix},
    \qquad
    \Theta_2 =
    \begin{bmatrix}
     0&      0  &   -1\\
     0&     0   &  0\\
     1&    0   &  0
    \end{bmatrix},
    \qquad
    \Theta_3 =
    \begin{bmatrix}
     0&      1  &   0\\
     -1&     0   &  0\\
     0&    0   &  0
    \end{bmatrix}.
\end{equation}
Note that finite-dimensional Hilbert spaces correspond to finite-level quantum systems such as the electron spin in an electromagnetic field (the physical setting from which the Pauli matrices originate).

\section{Quasilinear quantum stochastic dynamics}
\label{sec:QSS}

The open quantum system being considered interacts with a multichannel external bosonic field. This input field is modelled by an even number $m$ of quantum Wiener processes $W_1(t), \ldots, W_m(t)$ which are time-varying self-adjoint operators on a symmetric Fock space $\fF$ with a filtration $(\fF_t)_{t\> 0}$. These processes are assembled into a vector $W: = (W_k)_{1\< k\< m}$ whose future-pointing increments satisfy the quantum Ito relations
\begin{equation}
\label{dWdW_Omega_J_bJ}
    \rd W\rd W^{\rT}
    :=
    \Omega \rd t,
    \qquad
    \Omega
    :=
    (\omega_{jk})_{1\< j,k\< m}
    =
    I_m + iJ,
    \qquad
        J
        :=
       {\begin{bmatrix}
           0 & I_{m/2}\\
           -I_{m/2} & 0
       \end{bmatrix}}.
\end{equation}
As opposed to the identity diffusion matrix $I_m$ of the standard Wiener process \cite{KS_1991} in $\mR^m$, the quantum Ito matrix $\Omega$ is  a complex positive semi-definite Hermitian matrix. Its  imaginary part $\Im \Omega=J$  is an orthogonal antisymmetric matrix (so that $J^2 =-I_m$). This property is related to the fact that the quantum Wiener processes  $W_1, \ldots, W_m$ do not commute with each other  and have a nonzero two-point commutator matrix
\begin{equation}
\label{WWst}
    [W(s), W(t)^{\rT}]
     = 2i\min(s,t)J ,
    \qquad
    s,t\>0.
\end{equation}
The system-field interaction 
produces the output fields $Y_1(t), \ldots, Y_m(t)$ 
which are time-varying self-adjoint operators on the tensor-product Hilbert space $\fH:= \fH_0 \ox \fF$ as a common domain for the system and field operators, where $\fH_0$ is a Hilbert space for the action of the initial system variables $X_1(0), \ldots, X_n(0)$.
The Heisenberg evolution  of the vectors $X$ and     $Y   :=
    (Y_k)_{1\< k\< m}$ of the system variables and the output field variables is governed
by the Markovian Hudson-Parthasarathy QSDEs
\cite{HP_1984,P_1992}
\begin{align}
\label{dX}
    \rd X
    & =
    \cG(X)\rd t  - i[X,L^{\rT}]\rd W,\\
\label{dY}
  \rd Y & = 2JL\rd t + \rd W.
\end{align}
Here, the vector $L:= (L_k)_{1\< k \< m}$ is formed from self-adjoint system-field coupling operators $L_1, \ldots, L_m$ 
 on the space $\fH$, and hence, the dispersion $(n\x m)$-matrix $-i [X, L^{\rT}]$ also consists of self-adjoint operators on $\fH$.
The
drift vector
$    \cG(X)
$
in the QSDE (\ref{dX}) is obtained by the entrywise application of
the Gorini-Kossakowski-Sudar\-shan-Lindblad (GKSL) generator    \cite{GKS_1976,L_1976}, which  acts on a system operator $\xi$ (a function of the system variables)  as
\begin{equation}
\label{cG}
\cG(\xi)
   := i[H,\xi]
     +
     \frac{1}{2}
    \sum_{j,k=1}^m
    \omega_{jk}
    ( [L_j,\xi]L_k + L_j[\xi,L_k]).
\end{equation}
The system Hamiltonian $H$, which is also a self-adjoint operator on $\fH$, and the coupling operators $L_1, \ldots, L_m$ are functions (for example, polynomials) of the system variables $X_1, \ldots, X_n$ and inherit time dependence from them.
The superoperator $\cG$ is a quantum analogue of the infinitesimal generators of classical Markov processes \cite{KS_1991} 
and specifies the drift of the QSDE
\begin{equation}
\label{dxi}
    \rd \xi
    =
    \cG(\xi)\rd t - i[\xi,L^{\rT} ]\rd W,
\end{equation}
with $\xi(0)$ acting on the initial space $\fH_0$.
The structure of the generator (\ref{cG}) and the diffusion term in (\ref{dxi}), which are specified by the energy operators $H$, $L_1, \ldots, L_m$,   originate from the evolution of the system operator $\xi$:
\begin{equation}
\label{xiuni}
    \xi(t)
    =
    U(t)^{\dagger} (\xi(0)\ox \cI_{\fF}) U(t),
    \qquad
    t\> 0.
\end{equation}
Here, $U(t)$ is a unitary operator, which captures the system-field interaction over the time interval $[0,t]$, acts effectively on the subspace
\begin{equation}
\label{fHt}
    \fH_t:= \fH_0\ox \fF_t
\end{equation}
of the system-field space $\fH$  and satisfies the QSDE
\begin{equation}
\label{dU}
    \rd U(t)
    =
    -U(t) \Big(i(H(t)\rd t + L(t)^{\rT} \rd W(t)) + \frac{1}{2}L(t)^{\rT}\Omega L(t)\rd t\Big),
\end{equation}
with the initial condition $U(0)=\cI_{\fH}$. 
Under the quantum stochastic flow (\ref{xiuni}), the system variables evolve as
\begin{equation}
\label{uni}
    X(t)
    =
    U(t)^{\dagger} (X(0)\ox \cI_{\fF}) U(t),
\end{equation}
while (\ref{dY})  corresponds to the action of the flow on the output field variables in a   modified form:
\begin{equation}
\label{Y}
    Y(t)
    =
    U(t)^{\dagger}(\cI_{\fH_0}\ox W(t))U(t).
\end{equation}
Since $ \zeta\mapsto U(t)^{\dagger} \zeta U(t)$ is a unitary similarity transformation of  operators $\zeta$ on the system-field space $\fH$, 
the flow (\ref{uni}) preserves the algebraic property (\ref{XXX1}) and the structure constants:
\begin{align}
\nonumber
    \Xi(t)
    & =
    U(t)^{\dagger} (X(0)\ox \cI_{\fF})
    \overbrace{U(t)
    U(t)^{\dagger}}^{\cI_\fH} (X(0)^\rT\ox \cI_{\fF}) U(t)    \\
\nonumber
    & = U(t)^{\dagger} (\Xi(0)\ox \cI_{\fF}) U(t)
    =
    U(t)^{\dagger} ((\alpha + \beta \cdot X(0))\ox \cI_{\fF}) U(t)    \\
\label{XXuni}
    & =
    \alpha +
    \beta\cdot
    U(t)^{\dagger} (X(0)\ox \cI_{\fF}) U(t) = \alpha + \beta\cdot X(t),
    \qquad
    t\> 0.
\end{align}
The algebraic structure preservation (\ref{XXuni}) employs only the unitary nature of the quantum stochastic flow (\ref{xiuni}) and holds regardless of a particular form of the Hamiltonian $H$ and the coupling operators $L_1, \ldots, L_m$ which specify the QSDEs (\ref{dX}), (\ref{dxi}) involving the generator (\ref{cG}).

Furthermore, irrespective of a particular form of (\ref{dU}),  the unitary similarity transformation in (\ref{uni}), (\ref{Y}) preserves the commutativity between future system variables and past output field variables:
\begin{equation}
\label{XY}
        [X(t),Y(s)^{\rT}]
     =
    0,
    \qquad
    t\> s\> 0.
\end{equation}
However,
the output field variables $Y_1, \ldots, Y_m$ do not commute with each other 
and, in view of (\ref{dY}),
inherit the two-point CCRs (\ref{WWst}) from the input fields:
\begin{equation}
\label{YYst}
    [Y(s), Y(t)^{\rT}] = 2i\min(s,t)J,
    \qquad
    s,t\>0.
\end{equation}

In what follows, due to the reduction of polynomial (and more general) functions of the system variables to affine functions, mentioned in Section~\ref{sec:var}, we will be concerned with the case of \cite[Theorem 6.1]{EMPUJ_2016} when the energy operators $H$, $L_1, \ldots, L_m$ 
are affine functions of the system variables $X_1, \ldots, X_n$:
\begin{equation}
\label{H_LM}
    H
    =
    E^\rT X,
    \qquad
    L
    =
    MX + N,
\end{equation}
specified by an energy parameter $E\in \mR^n$ and coupling parameters $M \in \mR^{m\x n}$, $N \in \mR^m$.  An additive constant term for the Hamiltonian $H$ is omitted because $H$ enters the GKSL generator (\ref{cG}) only through the commutator, which makes the contribution  of such a constant vanish.
Although the assumption (\ref{H_LM}) is irrelevant for the validity of 
(\ref{XXuni}), 
its significance is that it leads to a quasilinear QSDE for the system variables with tractable moment dynamics (cf. \cite[Section~5]{VP_2012c}), which will be discussed in the subsequent sections. The following theorem, which  is similar to \cite[Lemma~4.2 and Theorem~6.1]{EMPUJ_2016}, is provided here for completeness along with a self-contained proof.

\begin{theorem}
\label{th:QSDE}
The QSDEs (\ref{dX}), (\ref{dY}) for the open quantum system with the Hamiltonian and coupling operators (\ref{H_LM}) and the dynamic variables satisfying (\ref{XXX}) along with (\ref{herm}), (\ref{Imalpha0}), take the form
\begin{align}
\label{dX1}
  \rd X  & = (AX + b) \rd t + B(X)\rd W,\\
\label{dY1}
  \rd Y & = 2J(MX+N)\rd t + \rd W.
\end{align}
Here, $A \in \mR^{n\x n}$, $b \in \mR^n$ are a matrix and a vector of coefficients, and $B(X)$ is an $(n\x m)$-matrix  of self-adjoint operators, which depend linearly on the system variables:
\begin{align}
\label{A}
    A
    & :=
        2
        \Theta \diam (E + M^\rT JN)
    +
    2
    \sum_{\ell = 1}^n
    \Theta_\ell
    M^\rT
    (
        M\Theta_{\ell\bullet \bullet}
        +
        J M\Re \beta_{\ell\bullet \bullet}
    ),\\
\label{b}
    b
    & :=
    2
    \sum_{\ell = 1}^n
    \Theta_\ell
    M^\rT
    JM\alpha_{\bullet \ell},\\
\label{BX}
    B(X)
    & := 2(\Theta \cdot X)M^\rT,
\end{align}
where the array $\Theta$ is given by (\ref{Theta}), and use is made of its product (\ref{diam}) with the vector $E + M^\rT J N \in \mR^n$. 
 \hfill$\square$
\end{theorem}
\begin{proof}
The QSDE (\ref{dY1}) is obtained by substituting $L$ from (\ref{H_LM}) into (\ref{dY}). We will now derive (\ref{dX1}).
From the first equality in (\ref{H_LM}) and the CCRs (\ref{XCCRTheta}) in the case (\ref{Imalpha0}), it follows that
\begin{equation}
\label{iHX}
    i[H,X]
     = -i[X,H]
    = -i[X,X^\rT]E
     = 2(\Theta\cdot X)E
 =
    2\sum_{\ell=1}^n \Theta_\ell E X_\ell
     =
    2
    (\Theta \diam E)
    X
\end{equation}
in view of (\ref{diam}), (\ref{cdotdiam}).
The second equality in (\ref{H_LM}) and the same CCRs (\ref{XCCRTheta})  imply that
\begin{align}
\nonumber
    \sum_{j,k=1}^m
    \omega_{jk}
     [L_j,X]L_k
    & =
    -
    [X,L^\rT] \Omega L
    =
    -
    [X,X^\rT] M^\rT \Omega (MX+N)\\
\label{LXL}
    & =
    -2i
    (\Theta \cdot X) M^\rT \Omega (MX+N)
    =
    -2i
    \sum_{\ell = 1}^n
    \Theta_\ell M^\rT \Omega (M \Xi_{\ell \bullet}^\rT+NX_\ell),
\end{align}
where the quantum Ito matrix $\Omega$ from (\ref{dWdW_Omega_J_bJ}) is also used. Here, $\Xi_{j \bullet}$ is the $j$th row of the matrix $\Xi$ in (\ref{XXX1}), so that,
in view of (\ref{XXX}),
\begin{equation}
\label{XjX}
    X_j X
    =
    \Xi_{j \bullet}^\rT
    =
    \alpha_{\bullet j}
    +
    \beta_{j\bullet \bullet} X,
    \qquad
    j=1, \ldots, n,
\end{equation}
where $\alpha_{\bullet j} = \alpha_{j\bullet}^\rT \in \mR^n$ is the $j$th column of the real symmetric matrix $\alpha$ due to (\ref{herm}), (\ref{Imalpha0}), 
 and $\beta_{j\bullet\bullet}:= (\beta_{jk\ell})_{1\< k,\ell\< n} \in \mC^{n\x n}$ is an appropriate section of the array $\beta$.
Substitution of (\ref{XjX}) into (\ref{LXL}) leads to
\begin{equation}
\label{LXL1}
    \sum_{j,k=1}^m
    \omega_{jk}
     [L_j,X]L_k
    =
    -2i
    \sum_{\ell = 1}^n
    \Theta_\ell M^\rT \Omega (M(\alpha_{\bullet\ell }+ \beta_{\ell\bullet \bullet} X)+NX_\ell).
\end{equation}
By a similar reasoning,
\begin{align}
\nonumber
    \sum_{j,k=1}^m
    &\omega_{jk}
    L_j[X,L_k]
     =
    -(L^\rT \Omega [L,X^\rT])^\rT
    =
    -((X^\rT M^\rT+N^\rT) \Omega M[X,X^\rT])^\rT\\
\nonumber
    & =
    -2i
    ((X^\rT M^\rT+N^\rT) \Omega M (\Theta \cdot X))^\rT
    =
    -2i
    \Big(
        (X^\rT M^\rT+N^\rT) \Omega M \sum_{\ell =1}^n \Theta_\ell X_\ell
    \Big)^\rT\\
\nonumber
    & =
    2i
    \sum_{\ell = 1}^n
    \Theta_\ell M^\rT \Omega^\rT (M \Xi_{\bullet\ell}+NX_\ell)
    =
    2i
    \sum_{\ell = 1}^n
    \Theta_\ell M^\rT \overline{\Omega}
    (M (\alpha_{\bullet \ell}+\beta_{\bullet \ell \bullet}X) +NX_\ell)\\
\label{LXL2}
    & =
    2i
    \sum_{\ell = 1}^n
    \Theta_\ell M^\rT \overline{\Omega}
    (M (\alpha_{\bullet \ell}+ \overline{\beta_{\ell \bullet \bullet}}X)+ NX_\ell),
\end{align}
where $\Xi_{\bullet k}$ is the $k$th column of the matrix $\Xi$. 
Here, use is also made of the antisymmetry $\Theta_\ell^\rT = -\Theta_\ell$ together with the sections
$\beta_{\bullet k \bullet}:= (\beta_{jk\ell})_{1\< j,\ell\< n} \in \mC^{n\x n}$ of the array $\beta$, and the Hermitian property of the matrices $\Omega$ and 
 $\beta_1, \ldots, \beta_n$. It follows from (\ref{LXL1}), (\ref{LXL2}) that
\begin{align}
\nonumber
     \frac{1}{2}
    \sum_{j,k=1}^m&
    \omega_{jk}
    ( [L_j,X]L_k + L_j[X,L_k])\\
\nonumber
     & =
    -2
    \Re
    \Big(
    i
    \Big(
    \sum_{\ell = 1}^n
    \Theta_\ell M^\rT \Omega M
    \begin{bmatrix}
        \alpha_{\bullet\ell} & \beta_{\ell\bullet \bullet}
    \end{bmatrix}
    \begin{bmatrix}
      \cI_\fH\\
      X
    \end{bmatrix}
    +
    (\Theta \diam (M^\rT \Omega N)) X
    \Big)
    \Big)
    \\
\nonumber
    & =
    2
    \sum_{\ell = 1}^n
    \Theta_\ell
    M^\rT
    \Im
    (\Omega M    \begin{bmatrix}
        \alpha_{\bullet\ell} & \beta_{\ell\bullet \bullet}
    \end{bmatrix}
    )
    \begin{bmatrix}
      \cI_\fH\\
      X
    \end{bmatrix}
    +
    2
    (\Theta \diam (M^\rT J N)) X
\\
\label{dec}
    & =
    2
    \sum_{\ell = 1}^n
    \Theta_\ell
    M^\rT
    (JM\alpha_{\bullet \ell}
    +
    (
        M\Im\beta_{\ell\bullet \bullet}
        +
        J M\Re \beta_{\ell\bullet \bullet}
    )
    X)
    +
        2
    (\Theta \diam (M^\rT J N)) X,
\end{align}
where use is also made of the structure of $\Omega$ from (\ref{dWdW_Omega_J_bJ}).
In view of (\ref{cG}), a combination of (\ref{iHX}) with (\ref{dec}) allows the drift of the QSDE (\ref{dX}) to be computed as
\begin{align}
\nonumber
    \cG(X)
    = &
        2
        (\Theta \diam E) X
        +
    2
    \sum_{\ell = 1}^n
    \Theta_\ell
    M^\rT
    (JM\alpha_{\bullet \ell}
    +
    (
        M\Im\beta_{\ell\bullet \bullet}
        +
        J M\Re \beta_{\ell\bullet \bullet}
    )
    X)    \\
\label{cG1}
    & +
        2
    (\Theta \diam (M^\rT J N)) X
       =  AX + b,
\end{align}
where the matrix $A$ and the vector $b$ are given by (\ref{A}), (\ref{b}). In accordance with (\ref{BX}) and an intermediate step in (\ref{LXL}),
\begin{equation}
\label{iXL}
  -i[X,L^\rT]
  =
  -i[X,X^\rT]M^\rT = 2(\Theta\cdot X) M^\rT = B(X).
\end{equation}
Substitution of (\ref{cG1}), (\ref{iXL}) into (\ref{dX}) establishes  (\ref{dX1}).
\end{proof}

Similarly to the case of  linear QSDEs for open quantum harmonic oscillators, the quasilinearity of the QSDE (\ref{dX1}) gives rise to a specific structure of its solutions discussed below.

\begin{theorem}
\label{th:var}
Under the conditions of Theorem~\ref{th:QSDE}, the system variables, governed by (\ref{dX1}),  satisfy
\begin{equation}
\label{Xts}
    X(t) = E_{t,s}X(s) + \int_s^t E_{t,\tau}\rd \tau b,
    \qquad
    t \> s \> 0,
\end{equation}
where
\begin{equation}
\label{Ets}
  E_{t,s}
  :=
    \lexp
  \int_s^t
  (A\rd \tau + 2\Theta \diam (M^\rT \rd W(\tau)))
\end{equation}
is a leftwards time-ordered exponential, and use is made of (\ref{diam}), (\ref{Theta}),  (\ref{A}), (\ref{b}). \hfill$\square$
\end{theorem}
\begin{proof}
By combining (\ref{cdotdiam}) with (\ref{BX}) and the commutativity between the forward Ito increments of $W$ and adapted quantum processes (whereby $[\rd W(t), X(s)^\rT] = 0$ for any $t\> s\> 0$), it follows that
\begin{equation}
\label{BXW}
    B(X)\rd W
    =
    2(\Theta \cdot X)M^\rT \rd W
    =
    2(\Theta \diam (M^\rT \rd W))X.
\end{equation}
This allows (\ref{dX1}) to be represented as a nonhomogeneous linear QSDE with the quantum Wiener process $W$ in the coefficients:
\begin{equation}
\label{dX2}
    \rd X
    =
    (A\rd t + 2\Theta \diam (M^\rT \rd W))X
    +
    b\rd t.
\end{equation}
The exponential (\ref{Ets}) provides the  fundamental solution for the  homogeneous part of the QSDE (\ref{dX2}):
\begin{align}
\nonumber
    \rd_t E_{t,s}
    & =
    (A\rd t + 2\Theta \diam (M^\rT \rd W(t)))
    E_{t,s}
    =
    A E_{t,s}\rd t + 2(\Theta \cdot E_{t,s})M^\rT \rd W(t)\\
\label{dEts}
    & =
    AE_{t,s}\rd t + B(E_{t,s})\rd W(t),
    \qquad
    t \> s \> 0,
    \qquad
    E_{s,s} := I_n,
\end{align}
in view of (\ref{BXW}).
The relation (\ref{Xts}) can now be obtained from (\ref{dEts}) by using the variation of constants.
\end{proof}

In contrast to the usual matrix exponentials as fundamental solutions of linear ODEs with constant coefficients, the time-ordered exponential $E_{t,s}$ in (\ref{Ets})  is an $(n\x n)$-matrix of self-adjoint quantum variables which commute with operators on the system-field subspace $\fH_s$ as in (\ref{fHt}). This follows from the continuous tensor-product structure of the Fock space \cite{PS_1972} and leads to a two-point extension of the one-point CCRs (\ref{XCCRTheta}) for the system variables:
\begin{align*}
\nonumber
    [X(t), X(s)^\rT]
    & =
    [E_{t,s} X(s), X(s)^\rT] + \int_s^t [E_{t,\tau}b, X(s)^\rT]\rd \tau \\
\label{XXcommts}
    & =
    E_{t,s} [X(s), X(s)^\rT]
    =
    2i E_{t,s} (\Theta \cdot X(s)),
    \qquad
    t \> s\> 0.
\end{align*}
Here, use is made of (\ref{Xts}),  and the commutativity between the entries of $E_{t,\tau}$ and $X(s)$ for all $t\> \tau \> s\> 0$ is combined with the identities $[\xi \eta, \zeta^\rT] = \xi\eta \zeta^\rT - (\zeta (\xi \eta)^\rT)^\rT = \xi\eta \zeta^\rT - (\zeta \eta^\rT \xi^\rT)^\rT = \xi\eta \zeta^\rT - \xi (\zeta \eta^\rT)^\rT = \xi [\eta, \zeta^\rT]$ which hold for appropriately dimensioned matrix $\xi$ and vectors $\eta$, $\zeta$ of quantum variables such that the entries of $\xi$ commute with those of $\eta$, $\zeta$.

The quasilinearity of the QSDE (\ref{dX1}) and linearity of (\ref{dY1}) will also be used in the subsequent sections in order to study the evolution of moments of the system variables and apply it to a quantum filtering problem.

\section{Moment dynamics of the system variables}
\label{sec:mom}

The proof of Theorem~\ref{th:bound} employed the representation of the matrix $\bE \Xi$ of the second-order moments and the covariance matrix $\cov(X)$ for the system variables in (\ref{covX}) in terms of their mean values in (\ref{mu}) as a corollary of the algebraic property (\ref{XXX1}). This is closely related to the reduction of quadratic functions of the system variables to affine functions  in (\ref{XRX}).  The following theorem provides a similar closed-form reduction for a wider class of nonlinear functions of the system variables.

\begin{theorem}
\label{th:red}
Under the conditions (\ref{XXX}), (\ref{herm}), for any entire function $f$ and a vector $u\in \mC^n$,
\begin{equation}
\label{fred}
  f(u^\rT X)
  =
        \begin{bmatrix}
        1 & \bzero_n^\rT
    \end{bmatrix}
    f\left(
    \begin{bmatrix}
      0 & u^\rT \\
      \alpha u &  \beta \diam u
    \end{bmatrix}
    \right)
    \begin{bmatrix}
      \cI_\fH\\
      X
    \end{bmatrix},
\end{equation}
where $\bzero_n$ is the column-vector of $n$ zeros. \hfill$\square$
\end{theorem}
\begin{proof}
It follows from (\ref{XXX1}), (\ref{cdotdiam}) that
\begin{align}
\nonumber
    \begin{bmatrix}
      \cI_\fH\\
      X
    \end{bmatrix}
    u^\rT X
    & =
    \begin{bmatrix}
      u^\rT X\\
      XX^\rT u
    \end{bmatrix}
    =
        \begin{bmatrix}
      u^\rT X\\
      (\alpha + \beta \cdot X) u
    \end{bmatrix}    \\
\nonumber
    & =
    \begin{bmatrix}
      u^\rT X\\
      \alpha u + (\beta \diam u) X
    \end{bmatrix}
    =
    \begin{bmatrix}
      0 & u^\rT \\
      \alpha u &  \beta \diam u
    \end{bmatrix}
    \begin{bmatrix}
      \cI_\fH\\
      X
    \end{bmatrix},
    \qquad
    u \in \mC^n.
\end{align}
Hence, by induction,
\begin{equation}
\label{uXr}
    (u^\rT X)^r
     =
    \begin{bmatrix}
        1 & \bzero_n^\rT
    \end{bmatrix}
    \begin{bmatrix}
      \cI_\fH\\
      X
    \end{bmatrix}
    (u^\rT X)^r
    =
        \begin{bmatrix}
        1 & \bzero_n^\rT
    \end{bmatrix}
    \begin{bmatrix}
      0 & u^\rT \\
      \alpha u &  \beta \diam u
    \end{bmatrix}^r
    \begin{bmatrix}
      \cI_\fH\\
      X
    \end{bmatrix}
\end{equation}
for all $r=0,1,2,\ldots$.
Substitution of (\ref{uXr}) into the Maclaurin series  of an arbitrary entire function $f$ leads to
\begin{align*}
    f(u^\rT X)
    & =
    \sum_{r=0}^{+\infty}\frac{1}{r!} f^{(r)}(0) (u^\rT X)^r
    =
        \begin{bmatrix}
        1 & \bzero_n^\rT
    \end{bmatrix}
    \sum_{r=0}^{+\infty}\frac{1}{r!} f^{(r)}(0)
    \begin{bmatrix}
      0 & u^\rT \\
      \alpha u &  \beta \diam u
    \end{bmatrix}^r
    \begin{bmatrix}
      \cI_\fH\\
      X
    \end{bmatrix}    \\
    & =
        \begin{bmatrix}
        1 & \bzero_n^\rT
    \end{bmatrix}
    f\left(
    \begin{bmatrix}
      0 & u^\rT \\
      \alpha u &  \beta \diam u
    \end{bmatrix}
    \right)
    \begin{bmatrix}
      \cI_\fH\\
      X
    \end{bmatrix},
\end{align*}
which establishes (\ref{fred}).
\end{proof}

The nonlinearity on the left-hand side of (\ref{fred}) is reminiscent of those in the classical
Lur'e systems \cite{L_1951}, which were used for perturbation modelling in a quantum control context, for example,  in \cite{VP_2012b}. 
Due to the affine dependence of the right-hand side of (\ref{fred}) on $X$, its application to an exponential function $f(z):= \re^{iz}$ allows the quasi-characteristic function (QCF) 
of the system variables  to be expressed in terms of their mean values (\ref{mu}):
\begin{equation}
\label{QCF}
    \Phi(t,u)
    :=
    \bE \re^{iu^\rT X(t)}
    =
        \begin{bmatrix}
        1 & \bzero_n^\rT
    \end{bmatrix}
    \exp\left(
    i
    \begin{bmatrix}
      0 & u^\rT \\
      \alpha u &  \beta \diam u
    \end{bmatrix}
    \right)
    \begin{bmatrix}
      1\\
      \mu(t)
    \end{bmatrix},
    \quad
    t \> 0,\
    u \in \mR^n.
\end{equation}
The QCF $\Phi$ depends on time only through the mean vector $\mu$ whose evolution  is amenable to closed-form computation due to the affine dependence of the drift vector of the QSDE (\ref{dX1}) on the system variables. More precisely, it is assumed in what follows that the input fields are in the vacuum quantum state $\ups$ on the Fock space $\fF$, which can be  described in terms of their  quasi-characteristic functionals \cite{CH_1971,HP_1984,P_1992} as
\begin{equation}
\label{vac}
    \bE \re^{i\int_0^t u(s)^\rT\rd W(s) } = \re^{-\frac{1}{2}\int_0^t|u(s)|^2\rd s},
    \qquad
    t\> 0,
\end{equation}
for any locally square integrable function $u: \mR_+\to \mR^m$. Here,
the quantum expectation (\ref{bE}) is  over the system-field density operator
\begin{equation}
\label{rho}
    \rho := \varpi \ox \ups,
\end{equation}
which is the tensor product of the initial system state $\varpi$ and the vacuum field state $\ups$. Since (\ref{vac}) involves only the input fields, the averaging  there reduces to that over the vacuum state $\ups$ (indeed, $\bE \zeta = \Tr(\ups\zeta)$ for any operator $\zeta$ on the Fock space $\fF$).  In the case of vacuum input fields, the martingale part $B(X)\rd W$ of the QSDE (\ref{dX1}) does not contribute to the quantum average of its right-hand side, and the mean vector $\mu$ in (\ref{mu}) evolves as
\begin{equation}
\label{EXdot}
    \dot{\mu }= A\mu + b,
\end{equation}
where $\dot{(\ )}$ is the time derivative.  The solutions of the ODE (\ref{EXdot}) satisfy
\begin{equation}
\label{musol}
  \mu(t) = \re^{(t-s)A}\mu(s) + \Psi(t-s)b,
  \qquad
  t \> s\> 0,
\end{equation}
where the $\mR^{n\x n}$-valued function
\begin{equation}
\label{psi}
    \Psi(t):= \int_0^t \re^{sA}\rd s
    =
    A^{-1}(\re^{tA}-I_n),
    \qquad
    t \in \mR,
\end{equation}
is obtained by evaluating  \cite{H_2008} the entire function $\frac{1}{z}(\re^{tz}-1)$ of $z \in \mC$ (extended by continuity to $t$ at $z=0$) at the matrix $A$ from (\ref{A}). 
Therefore, the mean vector has the limit
\begin{equation}
\label{mu*}
  \mu_\infty := \lim_{t\to +\infty}\mu(t) = -A^{-1} b,
\end{equation}
provided the matrix $A$ in (\ref{A}) is Hurwitz. In special cases, one of which is discussed in Section~\ref{sec:Pauli}, there exist easily verifiable sufficient conditions for the Hurwitz property of $A$. Also, (\ref{mu*}) implies that the mean Hamiltonian of the system in (\ref{H_LM}) satisfies
$
    \lim_{t\to +\infty}\bE H(t) = E^\rT \mu_\infty
$. 
The invariant quantum state of the system can be represented in terms of the corresponding QCF (\ref{QCF}):
\begin{equation}
\label{QCF*}
    \Phi_\infty(u)
     :=
    \lim_{t\to +\infty}
    \Phi(t,u)
    =
        \begin{bmatrix}
        1 & \bzero_n^\rT
    \end{bmatrix}
    \exp\left(
    i
    \begin{bmatrix}
      0 & u^\rT \\
      \alpha u &  \beta \diam u
    \end{bmatrix}
    \right)
    \begin{bmatrix}
      1\\
      \mu_\infty
    \end{bmatrix},
    \qquad
    u \in \mR^n.
\end{equation}
A similar combination of Theorem~\ref{th:red} with the limit relation (\ref{mu*}) allows for computation of infinite-horizon asymptotic growth rates for a class of integral cost functionals:
\begin{equation}
\label{rate}
    \lim_{T \to +\infty}
    \Big(
    \frac{1}{T}
    \bE
    \int_0^T
    \sum_{k=1}^s
    f_k(u_k^\rT X(t))
    \rd t
    \Big)
    =
        \begin{bmatrix}
        1 & \bzero_n^\rT
    \end{bmatrix}
    \sum_{k=1}^s
    f_k\left(
    \begin{bmatrix}
      0 & u_k^\rT \\
      \alpha u_k &  \beta \diam u_k
    \end{bmatrix}
    \right)
    \begin{bmatrix}
      1\\
      \mu_\infty
    \end{bmatrix},
\end{equation}
where $f_1, \ldots, f_s$ are entire functions with real coefficients, and $u_1, \ldots, u_s \in \mR^n$.  Here, use is made of the convergence of the Cesaro mean to the same limit.   In the case of quadratic cost functionals, the asymptotic growth rate can be  obtained directly by averaging (\ref{XRX}):
\begin{equation}
\label{quadrate}
    \lim_{T \to +\infty}
    \Big(
    \frac{1}{T}
    \bE
    \int_0^T
    X(t)^\rT R X(t)
    \rd t
    \Big)
    =
    \bra R, \alpha\ket_\rF
    +
    \begin{bmatrix}
      \bra R, \beta_1\ket_\rF &
      \ldots &
      \bra R, \beta_n\ket_\rF
    \end{bmatrix}
      \mu_\infty.
\end{equation}
The QCF $\Phi$ in  (\ref{QCF}) and its steady-state version $\Phi_\infty$  in (\ref{QCF*}) pertain to quantum statistical properties of the system variables at the same point in time, and so also do the cost functionals in (\ref{rate}), (\ref{quadrate}).  Multi-point statistical properties of the system variables at (in general, different) moments of time $t_1, \ldots, t_q \> 0$ can be described in terms of the multilinear function
\begin{equation}
\label{muq}
    M_q(t_1, \ldots, t_q; u_1, \ldots, u_q)
    :=
    \bE \lprod_{k=1}^q u_k^\rT X(t_k),
    \qquad
    u_1, \ldots, u_q \in \mR^n,
\end{equation}
where $\lprod_{k=1}^q\zeta_k := \zeta_q \x \ldots \x \zeta_1$ is the leftward-ordered product of linear operators (the order of multiplication is essential in the noncommutative case), and $q = 1,2,3,\ldots$. The corresponding mixed moments of the system variables are  recovered from (\ref{muq}) as
\begin{equation}
\label{moms}
    \bE \lprod_{k=1}^q
    X_{j_k}(t_k)
    =
    \d_{u_{j_1 1}}\ldots \d_{u_{j_q q}} M_q(t_1, \ldots, t_q; u_1, \ldots, u_q),
    \quad
    1\< j_1, \ldots, j_q \< n,
\end{equation}
where the partial derivatives are over the entries of the vectors $u_k:= (u_{jk})_{1\< j\< n}  \in \mR^n$. In particular,
\begin{equation}
\label{M1}
    M_1(t_1;u_1) = u_1^\rT \mu(t_1).
\end{equation}
The following theorem is concerned with the moments (\ref{muq}) in the case when the instants $t_1, \ldots, t_q$ form a nondecreasing  sequence.

\begin{theorem}
\label{th:mom}
Suppose the conditions of Theorem~\ref{th:QSDE} are satisfied, so that the system variables are governed by the QSDE (\ref{dX1}). Also, suppose the system-field state is given by (\ref{rho}), with the input fields being in the vacuum state $\ups$. 
Then, for any $q=2,3,\ldots$ and any  moments of time $0\< t_1\<  \ldots \<  t_q$, the functions $M_q$ in (\ref{muq}) satisfy the second-order  recurrence equation
\begin{align}
\nonumber
    M_q(t_1, \ldots, t_q; u_1, \ldots, u_q)
     =&
    M_{q-1}(t_1, \ldots, t_{q-1}; u_1, \ldots,  u_{q-2},   u
    )\\
\label{munext}
    & +
    w
    M_{q-2}(t_1, \ldots, t_{q-2}; u_1, \ldots, u_{q-2}),\\
\label{u}
    u :=&
    ((\beta\diam u_{q-1})^\rT \re^{(t_q-t_{q-1})A^\rT}
    +u_{q-1}b ^\rT\Psi(t_q-t_{q-1})^\rT)u_q    ,\\
\label{w}
    w:= &
    u_q^\rT \re^{(t_q-t_{q-1})A} \alpha u_{q-1}
\end{align}
for all $    u_1, \ldots, u_q \in \mR^n$, where the initial conditions $M_0:= 1$ and (\ref{M1})
are used along with (\ref{psi}).
\hfill$\square$
\end{theorem}
\begin{proof}
From the structure of the time-ordered exponential $E_{t,s}$ in (\ref{Ets}) and the continuous tensor-product structure of the Fock space $\fF$ and the vacuum state $\ups$, it follows that
\begin{equation}
\label{EEE1}
    \bE (E_{t,\tau}\eta) = \bE E_{t,\tau} \bE \eta,
    \qquad
    t\> \tau \> s\> 0,
\end{equation}
for any quantum variable $\eta$ on the system-field subspace $\fH_s$ in (\ref{fHt}). The averaging,  applied to the QSDE in the initial value problem (\ref{dEts}),  and the fact that the martingale part $B(E_{t,s})\rd W(t)$ does not contribute to the average, lead to  the ODE $\d_t \bE E_{t,s} = A \bE E_{t,s}$, with $\bE E_{s,s} = I_n$, and hence,
\begin{equation}
\label{EEE2}
    \bE E_{t,s} = \re^{(t-s)A},
    \qquad
    t \> s\> 0.
\end{equation}
A combination of (\ref{EEE1}), (\ref{EEE2}) with (\ref{Xts}), (\ref{psi}) yields
\begin{align}
\nonumber
  \bE (X(t) \eta)
  & =
  \bE \Big(\Big(E_{t,s}X(s) + \int_s^t E_{t,\tau}\rd \tau b\Big)\eta\Big) \\
\label{EEE3}
    & =
    \re^{(t-s)A}\bE(X(s)\eta) + \Psi(t-s) b \bE \eta,
    \qquad
    t\> s\> 0,
\end{align}
for any quantum variable $\eta$ on the system-field subspace $\fH_s$. In particular, by applying (\ref{EEE3}), with
\begin{align}
\label{ttt}
    t
    & := t_q\> s:= t_{q-1}\> \ldots \> t_1\> 0,\\
\label{etazeta}
    \eta
    & := \lprod_{k=1}^{q-1} u_k^\rT X(t_k)
    =
    u_{q-1}^\rT  X(s) \zeta,
    \qquad
    \zeta
     := \lprod_{k=1}^{q-2} u_k^\rT X(t_k),
\end{align}
to (\ref{muq}), it follows that
\begin{align}
\nonumber
    M_q(t_1, \ldots, &t_q; u_1, \ldots, u_q)
    =
    u_q^\rT
    \bE
    (
    X(t)
    \eta
    )
    \\
\nonumber
    &=
    u_q^\rT
    \big(
    \re^{(t-s)A}
    \bE
    (
    X(s)
    \eta)
    +
    \Psi(t-s)
    b
    \bE \eta
    \big)\\
\nonumber
    &=
    u_q^\rT \re^{(t-s)A}
    \bE
    (
    X(s)X(s)^\rT
    u_{q-1}
    \zeta
    )
    +
    u_q^\rT
    \Psi(t-s)
    b
    \bE \eta
    \\
\nonumber
    &=
    u_q^\rT \re^{(t-s)A}
    \bE
    (
    (\alpha + \beta \cdot X(s))
    u_{q-1}
    \zeta
    )
    +
    u_q^\rT
    \Psi(t-s)
    b
    \bE \eta
    \\
\nonumber
    &=
    u_q^\rT \re^{(t-s)A}
    \bE
    (
    (\alpha u_{q-1} + (\beta\diam u_{q-1}) X(s))
    \zeta
    )
    +
    u_q^\rT
    \Psi(t-s)
    b
    \bE \eta
    \\
\label{munext1}
    &=
    u_q^\rT \re^{(t-s)A}\alpha u_{q-1} \bE \zeta
    +
    u_q^\rT
    (\re^{(t-s)A}(\beta\diam u_{q-1})
    +\Psi(t-s)b u_{q-1}^\rT)
    \bE
     (X(s)
    \zeta
    ),
\end{align}
where use is also made of (\ref{XXX1}). Now, in view of (\ref{ttt}), (\ref{etazeta}) and (\ref{muq}),
\begin{align}
\label{Ezeta}
    \bE \zeta
    & = M_{q-2}(t_1, \ldots, t_{q-2}; u_1, \ldots, u_{q-2}),\\
\label{EXzeta}
    u^\rT\bE (X(s)\zeta)
    & = M_{q-1}(t_1, \ldots, t_{q-1}; u_1, \ldots, u_{q-2},u),
    \qquad
    u \in \mR^n.
\end{align}
Substitution of (\ref{Ezeta}), (\ref{EXzeta}), with $u:=
    ((\beta\diam u_{q-1})^\rT \re^{(t-s)A^\rT}
    +u_{q-1}b ^\rT\Psi(t-s)^\rT)u_q$,   into (\ref{munext1}) leads to (\ref{munext})--(\ref{w}). 
\end{proof}

Note that (\ref{u}), (\ref{w}) specify bilinear functions of $u_{q-1}, u_q \in \mR^n$. In particular,  in the case of $q=2$, application of (\ref{M1})--(\ref{w}) yields
\begin{align}
\nonumber
    M_2(s, t; u_1, u_2)
     &=
    M_1(s; u) +
    w\\
\nonumber
     & = u^\rT \mu(s)     + w\\
\nonumber
    & =
    u_2^\rT
    (\re^{(t-s)A}(\beta\diam u_1)
    +
    \Psi(t-s)b u_1^\rT) \mu(s) + w,\\
\label{M2}
    & =
    u_2^\rT
    (\re^{(t-s)A}(\alpha + \beta\cdot \mu(s))
    +
    \Psi(t-s)b \mu(s)^\rT)
    u_1,
    \qquad
    t\> s\> 0,
\end{align}
for all $u_1, u_2 \in \mR^n$,
where use is also made of (\ref{cdotdiam}). In accordance with (\ref{moms}), the two-point second-order  moments of the system variables can be  recovered from (\ref{M2}) as
\begin{align}
\nonumber
    \bE(X(t)X(s)^\rT)
    & = \d_{u_1}\d_{u_2} M_2(s, t; u_1, u_2)\\
\label{EXX}
    &  =
    \re^{(t-s)A}(\alpha + \beta\cdot \mu(s))
    +
    \Psi(t-s)b \mu(s)^\rT.
\end{align}
A combination of (\ref{EXX}) with (\ref{musol}) leads to the quantum covariance function of the system variables:
\begin{align}
\nonumber
  \cov(X(t),X(s))
   := &
  \bE(X(t)X(s)^\rT) - \mu(t)\mu(s)^\rT\\
\nonumber
   = &
    \re^{(t-s)A}(\alpha + \beta\cdot \mu(s))
    +
    \Psi(t-s)b \mu(s)^\rT  \\
\nonumber
    & - (\re^{(t-s)A}\mu(s) + \Psi(t-s)b)\mu(s)^\rT\\
\label{covXX}
    = &
    \re^{(t-s)A}(\alpha + \beta\cdot \mu(s) - \mu(s)\mu(s)^\rT),
    \qquad
    t\> s\> 0.
\end{align}
In the invariant quantum state (which the system has if the matrix $A$ is Hurwitz, and the input fields are in the vacuum state, as discussed above),  the function (\ref{covXX}) depends only on the time difference:
\begin{equation}
\label{mho}
    \cov(X(t),X(s))
    =
    \left\{
    \begin{matrix}
      \re^{(t-s)A} \Gamma & {\rm if} &t\> s\\
      \Gamma \re^{(s-t)A^\rT}  & {\rm if} &s> t\\
    \end{matrix}
    \right.
    =:
    \mho(t-s),
    \qquad
    s,t\> 0.
\end{equation}
Here, similarly to (\ref{covX}),
\begin{equation}
\label{Gamma}
     \Gamma:= \alpha + \beta\cdot \mu_\infty - \mu_\infty\mu_\infty^\rT
\end{equation}
is the invariant quantum covariance matrix of the system variables. The corresponding spectral density $S: \mR\to \mC^{n\x n}$ is obtained by applying the Fourier transform to (\ref{mho}):
\begin{align}
\nonumber
    S(\omega)
    & :=
    \int_\mR \re^{-i\omega \tau} \mho(\tau)\rd \tau\\
\nonumber
    & =
    \int_0^{+\infty} \re^{- \tau(i\omega I_n-A)} \rd \tau \Gamma
    +
    \Gamma
    \int_{-\infty}^0 \re^{-\tau (i\omega I_n +A^\rT)} \rd \tau\\
\nonumber
    & =
    (i\omega I_n - A)^{-1} \Gamma - \Gamma (i\omega I_n + A^\rT)^{-1}\\
\nonumber
    & =
    (i\omega I_n - A)^{-1} (A\Gamma + \Gamma A^\rT) (i\omega I_n + A^\rT)^{-1} \\
\label{specden}
    & =
    -(i\omega I_n - A)^{-1} \Ups(i\omega I_n + A^\rT)^{-1}
    =
    S(\omega)^*\succcurlyeq  0,
    \qquad
    \omega \in \mR,
\end{align}
where $\Ups$ is a complex positive semi-definite Hermitian matrix of order $n$ given by the expectation \begin{equation}
\label{Ups}
    \Ups := \bE (B(X)\Omega B(X)^{\rT})
\end{equation}
over the invariant quantum state. In (\ref{specden}), we have also used the algebraic Lyapunov equation (ALE)
\begin{equation}
\label{GammaALE}
    A \Gamma + \Gamma A^\rT + \Ups = 0,
\end{equation}
which is obtained as follows.   
A combination of the QSDE (\ref{dX1}) with the ODE (\ref{EXdot}) implies that the centred quantum process
\begin{equation}
\label{Xc}
    \breve{X} := X - \mu
\end{equation}
satisfies the QSDE 
$    
    \rd \breve{X} = A\breve{X} + B(X)\rd W
$. 
From this QSDE and the quantum Ito lemma \cite{HP_1984,P_1992}, it follows that
\begin{align}
\nonumber
    \rd (\breve{X}\breve{X}^\rT)
    =&
    (\rd \breve{X}) \breve{X}^\rT + \breve{X} \rd \breve{X}^\rT + \rd \breve{X} \rd \breve{X}^\rT\\
\nonumber
    = &
    (A \breve{X} \rd t + B(X)\rd W)\breve{X}^\rT\\
\nonumber
    & +
    \breve{X}(\breve{X}^\rT A^\rT  \rd t + \rd W^\rT B(X)^\rT)\\
\nonumber
    & + B(X) \rd W \rd W^\rT B(X)^\rT\\
\nonumber
    =& (A \breve{X}\breve{X}^\rT + \breve{X}\breve{X}^\rT A^\rT + B(X) \Omega B(X)^\rT)\rd t\\
\label{dXX}
    & + B(X)(\rd W) \breve{X}^\rT +  \breve{X} (\rd W)^\rT B(X)^\rT,
\end{align}
where $\Omega$ is the Ito matrix of the quantum Wiener process $W$ from (\ref{dWdW_Omega_J_bJ}), and  the commutativity $[\rd W, X^\rT] = 0$ is also used. Since the input fields are assumed to be in the vacuum state, the last line of the QSDE (\ref{dXX}) describes its martingale part which does not contribute to the quantum average of the right-hand side. Hence, by averaging both sides of (\ref{dXX}),  it follows that the quantum covariance matrix $\cov(X) = \bE(\breve{X}\breve{X}^\rT)$ in (\ref{covX}), represented in terms of (\ref{Xc}),
satisfies the Lyapunov ODE
\begin{equation}
\label{covXdot}
    (\cov(X))^{^\centerdot}
    =   A \cov(X) + \cov(X) A^\rT + V(\mu). 
\end{equation}
Here, 
\begin{align}
\nonumber
    V(\mu)& := \bE (B(X)\Omega B(X)^\rT )\\
\nonumber
    & =
    4
    \bE((\Theta \cdot X)M^\rT \Omega M(\Theta \cdot X)^\rT)\\
\nonumber
    & =
    -4
    \sum_{j,k=1}^n
    \bE \Xi_{jk}
    \Theta_j
    M^\rT\Omega  M
    \Theta_k
    \\
\label{V}
    & =
    -4
    \sum_{j,k=1}^n
    \Big(
        \alpha_{jk}
        +
        \sum_{\ell=1}^n
        \beta_{jk\ell}\mu_\ell
    \Big)
    \Theta_j
    M^\rT\Omega M
    \Theta_k
\end{align}
is a complex positive semi-definite Hermitian matrix of order $n$, which  
is computed by combining (\ref{BX}) with (\ref{XXX}), (\ref{mu}), (\ref{Theta}) and depends on time through the mean vector $\mu$  for the system variables in  (\ref{musol}). The properties $V = V^* \succcurlyeq 0$ follow from  the fact that $B(X)$ consists of self-adjoint quantum variables and $\Omega = \Omega^* \succcurlyeq 0$. The invariant quantum covariance matrix $\Gamma$ in (\ref{Gamma})  is a steady-state solution of the ODE (\ref{covXdot}), with the matrix $V(\mu)$ in (\ref{V}) replaced  with its limit $\Ups = V(\mu_\infty)$ in (\ref{Ups}),  thus leading to the ALE (\ref{GammaALE}).

The recurrence relations of Theorem~\ref{th:mom}  for multi-point mixed moments are applicable to the development of methods for computing quadratic-exponential cost functionals and their growth rates. Similarly to open quantum harmonic oscillators \cite{VPJ_2018a}, such functionals can be employed as risk-sensitive robust performance criteria for quasilinear quantum systems.   However, this line of research is beyond the scope of the present paper, and a mean square  cost functional (as in (\ref{quadrate})) will be used  for a quantum filtering problem.

\section{Mean square optimal linear observer design}\label{sec:filt}

Consider the measure\-ment-based quantum filtering setup shown in Fig.~\ref{fig:filt}
\begin{figure}[htbp]
\unitlength=0.8mm
\linethickness{0.4pt}
\begin{picture}(130.00,25.00)
    \put(5,13){\makebox(0,0)[cc]{$W$}}
    \put(10,13){\vector(1,0){20}}
    \put(30,5){\framebox(20,15)[cc]{{}}}
    \put(30,7){\makebox(20,15)[cc]{{\small quantum}}}
    \put(30,3){\makebox(20,15)[cc]{{\small plant}}}
    \put(50,13){\vector(1,0){20}}
    \put(60,17){\makebox(0,0)[cc]{$Y$}}
    \put(70,5){\framebox(20,15)[cc]{{}}}
    \put(70,7){\makebox(20,15)[cc]{{\small measuring}}}
    \put(70,3){\makebox(20,15)[cc]{{\small device}}}
    \put(90,13){\vector(1,0){20}}
    \put(100,17){\makebox(0,0)[cc]{$Z$}}
    \put(110,5){\framebox(20,15)[cc]{{}}}
    \put(110,7){\makebox(20,15)[cc]{{\small classical}}}
    \put(110,3){\makebox(20,15)[cc]{{\small observer}}}
    \put(130,13){\vector(1,0){20}}
    \put(155,13){\makebox(0,0)[cc]{$\xi$}}
\end{picture}
\caption{A filtering setup for the quantum plant
with the output field $Y$, driven by the quantum Wiener process $W$ according to (\ref{dX1}), (\ref{dY1}), and a classical   linear observer with the measurement signal $Z$ at the input.}
\label{fig:filt}
\end{figure}
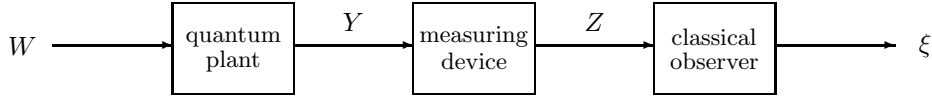
for a quantum plant modelled as in Theorem~\ref{th:QSDE}. The plant output field $Y$ is converted by a measuring device
to a  multichannel observation process $Z:=(Z_k)_{1\< k\< r}$
consisting of $r$ time-varying self-adjoint quantum variables, which are used by a classical linear observer in order to produce an $\mR^n$-valued estimate $\xi$ for the current plant variables. More precisely, the measurement is modelled by a static linear relation \cite{N_2014}
\begin{equation}
\label{ZY}
    Z
     = D Y,
\end{equation}
where $D\in \mR^{r\x m}$ is a constant matrix with $r\< \frac{m}{2}$ rows satisfying the conditions
\begin{equation}
\label{DD_DJD}
  F:= DD^{\rT}  \succ 0,
  \qquad
  DJD^{\rT}  = 0,
\end{equation}
the first of which is equivalent to $D$ being  of full row rank.
It follows from (\ref{ZY}) that the process $Z$ inherits the nondemolition property \cite{B_1983,B_1989}
\begin{equation}
\label{XZ}
        [X(t),Z(s)^{\rT}]
     =
     [X(t),Y(s)^{\rT}]D^\rT
     =
    0,
    \qquad
    t\> s\> 0,
\end{equation}
from the plant output field $Y$ in (\ref{XY}).  In view of (\ref{YYst}), the second condition in (\ref{DD_DJD}) implies that
\begin{equation}
\label{ZZst}
    [Z(s), Z(t)^{\rT}]
    =
    D[Y(s), Y(t)^{\rT}]D^\rT
    =
    2i\min(s,t)DJD^\rT
    =0,
    \qquad
    s,t\> 0.
\end{equation}
The relations (\ref{XZ}), (\ref{ZZst}) justify $Z$ as a nondemolition measurement process, whose
entries  $Z_1, \ldots, Z_r$ commute with future plant variables and between themselves at all times (and hence, are accessible to simultaneous continuous measurement). The process $Z$ is organised as a classical  Ito process \cite{KS_1991} with values in  $\mR^r$ and a real positive definite diffusion matrix
\begin{equation}
\label{DOD}
    D\Omega  D^\rT =DD^{\rT} +iDJD^{\rT}=F
\end{equation}
in view of (\ref{dWdW_Omega_J_bJ}), (\ref{DD_DJD}), so that
 $
    \rd Z \rd Z^\rT = D \rd Y \rd Y^\rT D^\rT= D \rd W \rd W^\rT D^\rT = D\Omega D^\rT \rd t = F \rd t
 $.
For any time $t\>0$, an operator-valued extension $f(X(t))$ of a  complex-valued function $f$ to the plant variables and the past observation history
\begin{equation}
\label{Zt}
    \fZ_t
    :=
    \{
        Z_1(s), \ldots, Z_r(s):\ 0\<s\< t
    \}
\end{equation}
form a set of pairwise commuting (and hence, compatible) quantum variables.
The commutative von Neumann algebra $\cZ_t$,  generated by the past observation history $\fZ_t$  in (\ref{Zt}), is the information available to the causal classical observer in Fig.~\ref{fig:filt} at time $t$.

We will now take into account the dynamics of the observation process $Z$ in (\ref{ZY}), which, in view of (\ref{dY1}), satisfies  the QSDE
\begin{equation}
\label{dZ}
  \rd Z = D \rd Y = (CX + d) \rd t + D \rd W,
\end{equation}
where $C\in \mR^{r\x n}$, $d\in \mR^r$ are related to the plant-field coupling parameters $M$, $N$  from (\ref{H_LM}) by
\begin{equation}
\label{Cd}
  C:= 2DJM,
  \qquad
  d := 2DJ N.
\end{equation}
This allows $Z$ to be used as an input to a Luenberger observer \cite{AM_1979}, whose state $\xi$  is a classical $\mR^n$-valued Ito process with respect to the filtration $\cZ:= (\cZ_t)_{t\> 0}$  governed by an SDE 
\begin{equation}
\label{dXhat}
  \rd \xi
  =
  (A\xi + b) \rd t + K (\rd Z- (C\xi + d)\rd t)
  =
    ((A-KC)\xi + b-Kd) \rd t + K \rd Z,
\end{equation}
initialised at the quantum average of the initial plant variables:
\begin{equation}
\label{xi0}
    \xi(0) := \bE X(0) = \mu(0).
\end{equation}
The observer is specified by a gain matrix $K \in \mR^{n\x r}$ which is a given continuous function of time,   not constrained by quantum physical realizability conditions in contrast to the coefficients of the plant QSDEs (\ref{dX1}), (\ref{dY1}). Although  $\xi(t)$ differs from the conditional expectation $\bE(X(t)\mid \cZ_t)$,  the quantity
\begin{equation}
\label{nov}
    \rd Z- (C\xi + d)\rd t
    =  (CX + d) \rd t + D \rd W - (C\xi + d)\rd t
     =
    Ce\rd t + D\rd W
\end{equation}
in (\ref{dXhat}) is similar to the innovation process increment in the Kalman filter \cite{LS_2001}. Here, (\ref{dZ}) is used along with the ``estimation error'' process
\begin{equation}
\label{e}
  e:= (e_k)_{1\< k \< n} :=  X-\xi
  =
  \begin{bmatrix}
    I_n & -I_n
  \end{bmatrix}
  \cX,
\end{equation}
which consists of $n$ time-varying self-adjoint quantum variables and is represented in terms of
the augmented vector
\begin{equation}
\label{cX}
  \cX
  :=
  {\begin{bmatrix}
    X\\
    \xi
  \end{bmatrix}}
\end{equation}
of the plant and observer variables. These dynamic variables satisfy the one-point CCRs
\begin{equation}
\label{cXCCR}
    [\cX, \cX^\rT]
    =
    \begin{bmatrix}
      [X,X^\rT] & [X,\xi^\rT] \\
      [\xi,X^\rT] & [\xi,\xi^\rT]
    \end{bmatrix}
    =
    \begin{bmatrix}
      2i \Theta \cdot X & 0\\
      0 & 0
    \end{bmatrix}
\end{equation}
and the QSDE
\begin{equation}
\label{dcX}
    \rd \cX =
    \Big(
        \begin{bmatrix}
          A & 0 \\
          KC & A-KC
        \end{bmatrix}
        \cX
        +
        \begin{bmatrix}
          b\\
          b
        \end{bmatrix}
    \Big)
    \rd t
    +
    \begin{bmatrix}
      B(X)\\
      KD
    \end{bmatrix}
    \rd W
\end{equation}
which is obtained by combining (\ref{dX1}), (\ref{dXhat}), (\ref{nov}) with (\ref{cX}).
In view of the identity
$$
      \begin{bmatrix}
    I_n & -I_n
  \end{bmatrix}
            \begin{bmatrix}
          A & 0 \\
          KC & A-KC
        \end{bmatrix}
        =
        (A-KC)
  \begin{bmatrix}
    I_n & -I_n
  \end{bmatrix} ,
$$
it follows from (\ref{dcX}) that the process $e$ in  (\ref{e}) satisfies the QSDE
\begin{equation}
\label{de}
  \rd e
  =
    \begin{bmatrix}
    I_n & -I_n
  \end{bmatrix}
  \rd \cX
  =
  (A-KC) e \rd t + (B(X)-KD)\rd W.
\end{equation}
Since the input quantum fields are in the vacuum state, the averaging of the QSDE (\ref{de}) leads to the ODE
$
    (\bE e)^{^\centerdot} = (A-KC) \bE e
$,
with the initial condition $\bE e(0) = \bE X(0) - \mu(0) = 0$ in view of (\ref{xi0}). Therefore, $\bE e(t)=0$ for all $t\> 0$, so that $\xi$ is an unbiased estimator for the vector $X$ of the plant variables.
Furthermore, (\ref{cXCCR}) implies that $e$ inherits 
the one-point CCRs (\ref{XCCRTheta}):
\begin{equation}
\label{eecomm}
    [e,e^\rT]
    =
    \begin{bmatrix}
    I_n & -I_n
  \end{bmatrix}
  [\cX,\cX^\rT]
      \begin{bmatrix}
    I_n \\ -I_n
  \end{bmatrix}=
    [X,X^\rT] = 2i\Theta \cdot X,
\end{equation}
which do not depend on the gain matrix $K$ in (\ref{dXhat}). The averaging of (\ref{eecomm}) allows the quantum covariance matrix of $e$ to be split as
\begin{equation}
\label{Eee}
    G:= \cov(e) = \bE(ee^\rT) = P + \frac{1}{2} \bE[e,e^\rT] = P + i\Theta \cdot \mu,
\end{equation}
where
\begin{equation}
\label{P}
    P:= \Re G.
\end{equation}
Note that $G$ in (\ref{Eee}) is a positive semi-definite Hermitian  matrix whose antisymmetric imaginary part does not contribute to
\begin{equation}
\label{SPS}
    \Tr (SGS^\rT)  = \Tr (SPS^\rT). 
\end{equation}
This trace quantifies the mean square error for the vector $S \xi$ as an unbiased estimator of the vector $SX$ related to the plant variables by  a fixed but otherwise arbitrary weighting matrix $S \in \mR^{\nu \x n}$. In particular, 
\begin{equation}
\label{Ptrace}
    \Tr P  = \Tr \bE (ee^\rT) = \bE \sum_{k=1}^n e_k^2.
\end{equation}

\begin{theorem}
\label{th:PALE}
Suppose the quantum plant is described by Theorem~\ref{th:QSDE}, the   plant-field state is given by (\ref{rho}), and the input quantum fields are in the vacuum state $\ups$. Also, let the  Luenberger observer be specified by (\ref{ZY})--(\ref{xi0}). Then the matrix $P$ in (\ref{P}) for the estimation error $e$ in (\ref{e}) satisfies the Lyapunov ODE
\begin{equation}
\label{Pdot}
    \dot{P}
    = (A-KC)P + P(A-KC)^\rT + \Sigma(\mu)  - KD B(\mu)^\rT - B(\mu)D^\rT K^\rT + K F K^\rT,
\end{equation}
initialised at
\begin{equation}
\label{P0}
  P(0) = \Re \cov(X(0)) = \alpha + \Re \beta\cdot \mu(0) - \mu(0)\mu(0)^\rT.
\end{equation}
Here,
\begin{equation}
\label{Sigma0}
    \Sigma(\mu) :=
    -4
    \sum_{j,k=1}^n
    \Theta_j
    M^\rT
    \Big(
        \alpha_{jk} I_m
        +
        \sum_{\ell=1}^n
        \mu_\ell
        (\Re \beta_{jk\ell} I_m - \theta_{jk\ell}J)
    \Big)
    M
    \Theta_k
\end{equation}
is a real positive semi-definite symmetric matrix, which depends on time through the mean vector $\mu$  for the plant variables in  (\ref{musol}). \hfill$\square$
\end{theorem}
\begin{proof}
Similarly to (\ref{dXX}), 
a combination of the quantum Ito lemma 
with the QSDE (\ref{de}) leads to 
\begin{align}
\nonumber
    \rd (ee^\rT)
    =&
    (\rd e) e^\rT + e \rd e^\rT + \rd e \rd e^\rT\\
\nonumber
    = &
    ((A-KC) e \rd t + (B(X)-KD)\rd W)e^\rT\\
\nonumber
    & +
    e(e^\rT (A-KC)^\rT  \rd t + \rd W^\rT (B(X)-KD)^\rT)\\
\nonumber
    & + (B(X)-KD) \rd W \rd W^\rT (B(X)-KD)^\rT\\
\nonumber
    =& ((A-KC) ee^\rT + ee^\rT(A-KC)^\rT + (B(X)-KD) \Omega (B(X)-KD)^\rT)\rd t\\
\label{dee}
    & + (B(X)-KD)(\rd W) e^\rT +  e (\rd W)^\rT (B(X)-KD)^\rT.
\end{align}
Since the input fields are in the vacuum state, the martingale part on the last line of the QSDE (\ref{dee}) does not contribute to its averaging, which yields the following Lyapunov ODE for the matrix $G$ in (\ref{Eee}):
\begin{align}
\nonumber
    \dot{G}
    = &  (A-KC)G + G(A-KC)^\rT + \bE ((B(X)-KD)\Omega (B(X)-KD)^\rT )\\
\nonumber
    = &  (A-KC)G + G(A-KC)^\rT + \bE (B(X)\Omega B(X)^\rT)\\
\label{Gdot}
     & - KD\Omega B(\mu)^\rT - B(\mu)\Omega D^\rT K^\rT + K F K^\rT,
\end{align}
due to (\ref{DOD}). We have also used the linearity of the map $B$ in (\ref{BX}), leading to $\bE B(X) = B(\mu) \in \mR^{n\x m}$ in view of (\ref{mu}).
 The ODE (\ref{Pdot}) for the matrix $P$ in (\ref{P}) is obtained by taking the real part of (\ref{Gdot}) and using 
$\Re \Omega = I_m$ from (\ref{dWdW_Omega_J_bJ}) along with the matrix
\begin{align}
\nonumber
    \Sigma(\mu)
    & := \Re V(\mu)
    \\
\nonumber
    & =
    -4\Re
    \sum_{j,k=1}^n
    \Big(
        \alpha_{jk}
        +
        \sum_{\ell=1}^n
        \beta_{jk\ell}\mu_\ell
    \Big)
    \Theta_j
    M^\rT\Omega M
    \Theta_k\\
\label{Sigma}
    & =
    -4
    \sum_{j,k=1}^n
    \Theta_j
    M^\rT
    \Big(
        \alpha_{jk} I_m
        +
        \sum_{\ell=1}^n
        \mu_\ell
        (\Re \beta_{jk\ell} I_m - \theta_{jk\ell}J)
    \Big)
    M
    \Theta_k,
\end{align}
which is associated with (\ref{V}) and satisfies $\Sigma = \Sigma^\rT \succcurlyeq 0$ as the real part of $V = V^*\succcurlyeq 0$. 
The relation (\ref{Sigma}) establishes (\ref{Sigma0}). The initial condition (\ref{P0}) follows from (\ref{xi0}), (\ref{covX}).
\end{proof}

The time-varying gain matrix $K$ plays the role of a free parameter
which can be chosen so as to minimise the solution $P$ of the Lyapunov ODE (\ref{Pdot}) in the sense of the positive semi-definite matrix ordering, thereby minimising the quantity (\ref{SPS}) (for any fixed but otherwise arbitrary weighting matrix $S$) and its particular version (\ref{Ptrace}). To this end, note that the matrix  $\Sigma(\mu)$ in (\ref{Sigma0}) does not depend on $K$, and hence, the right-hand side of (\ref{Pdot})  is a quadratic function of $K$, which admits a completion of the square:
\begin{align}
\nonumber
    \dot{P}
    =&
    AP + PA^\rT + \Sigma(\mu)\\
\nonumber
    & -K (CP + DB(\mu)^\rT) - (PC^\rT + B(\mu)D^\rT)K^\rT +KF K^\rT\\
\nonumber
    =&
    AP + PA^\rT + \Sigma(\mu) - K_* F K_*^\rT\\
\nonumber
    & +
    (K-K_*)F (K-K_*)^\rT\\
\label{PdotK}
    \succcurlyeq &
    AP + PA^\rT + \Sigma(\mu) - K_* F K_*^\rT,
\end{align}
where
\begin{equation}
\label{K*}
    K_*
    :=
      (PC^\rT + B(\mu)D^\rT)F^{-1}.
\end{equation}
By a standard monotonicity argument, the minimal  solution $P$ corresponds to $K$ which  minimises the right-hand side of (\ref{PdotK}) at every moment of time. Since $F\succ 0$ in view of (\ref{DD_DJD}),   the minimum is achieved only at $K= K_*$. The resulting Riccati ODE
\begin{equation}
\label{RODE}
  \dot{P} = AP + PA^\rT + \Sigma(\mu) - (PC^\rT + B(\mu)D^\rT)F^{-1}(CP + DB(\mu)^\rT),
\end{equation}
with the initial condition $P(0)$ in  (\ref{P0}) and $\mu$, $\Sigma$ computed according to (\ref{musol}), (\ref{Sigma}), defines the optimal gain matrix (\ref{K*}) for the Luenberger observer (\ref{dXhat}). In the case when $A$ is Hurwitz, the steady-state version of the mean square optimal  observer is obtained by substituting $\mu_\infty$ from (\ref{mu*}) into (\ref{Sigma}) and finding a unique stabilising solution $P_\infty$ of the algebraic Riccati equation
\begin{equation}
\label{ARE}
  AP_\infty + P_\infty A^\rT + \Sigma(\mu_\infty) - (P_\infty C^\rT + B(\mu_\infty)D^\rT)F^{-1}(CP_\infty + DB(\mu_\infty)^\rT) = 0,
\end{equation}
which yields the steady-state gain matrix
\begin{equation}
\label{Kinf}
    K_\infty
    :=
      (P_\infty C^\rT + B(\mu_\infty)D^\rT)F^{-1}.
\end{equation}
The stabilising property is understood in the usual sense \cite{AM_1979,LS_2001} that the matrix $A-K_\infty  C$ is Hurwitz. In view of the structure of the Riccati equations and the gain matrices in (\ref{K*})--(\ref{Kinf}),  the corresponding Luenberger observer (\ref{dXhat}) is similar to the classical Kalman filter. This steady-state filtering regime uses the Hurwitz property of the matrix $A$ from (\ref{A}), whose verification is illustrated by the following example.

\section{Example: Pauli matrices as initial quantum plant variables}\label{sec:Pauli}

Suppose the quantum plant has $n=3$ dynamic variables, which are organised initially as the Pauli matrices (\ref{X123}):
\begin{equation}
\label{XPauli}
    X_k(0) = \sigma_k,
    \qquad
    k = 1,2,3,
\end{equation}
on the Hilbert space $\fH_0:= \mC^2$. Although the   quantum stochastic flow (\ref{uni}) complicates  the nature of the subsequent plant variables, they retain the algebraic structure of the Pauli  matrices over the course of time in accordance with  (\ref{XXuni}). 
Substitution of the corresponding structure constants, specified at the end of Section~\ref{sec:var}, into (\ref{A}) yields
\begin{equation}
\label{APauli}
    A
     =
        2
        \Theta \diam (E + M^\rT JN)
    +
    2
    \sum_{\ell = 1}^3
    \Theta_\ell
    M^\rT
        M
            \Theta_\ell,
\end{equation}
where $M \in \mR^{m\x 3}$ is the plant-field coupling matrix from (\ref{H_LM}).  Here,
use is made of the invariance of the Levi-Civita symbol $\eps_{jk\ell}$ under cyclic permutations of its indices, whereby the matrices (\ref{T123}) satisfy
$$
    \Theta_\ell = \Theta_{\ell\bullet \bullet},
    \qquad
    \ell = 1,2,3.
$$
The relation of the array $\Theta$ to the Levi-Civita symbol in this case implies that the first and second terms in (\ref{APauli}) are the antisymmetric and symmetric parts of the matrix $A$, respectively, and hence,
\begin{equation}
\label{AA}
    A+A^\rT =     4
    \sum_{\ell = 1}^3
    \Theta_\ell
    M^\rT
        M
            \Theta_\ell.
\end{equation}
The antisymmetry of the matrix $
        \Theta \diam (E + M^\rT JN)
$ in (\ref{APauli}) follows from the identity
$$
    \Theta \diam u
    =
    \begin{bmatrix}
      \Theta_1 u &
      \Theta_2 u &
    \Theta_3 u
    \end{bmatrix}
    =
    \begin{bmatrix}
      0 & -u_3 & u_2\\
      u_3 & 0 & -u_1\\
      -u_2 & u_1 & 0
    \end{bmatrix},
    \qquad
    u:= (u_k)_{1\< k\< 3} \in \mR^3,
$$
in view of (\ref{diam}), (\ref{T123}),
due to which $(\Theta\diam u)v$ is the cross product of vectors $u,v\in \mR^3$. Furthermore,
\begin{equation}
\label{MM}
    \Lambda
    :=
        -\sum_{\ell = 1}^3
    \Theta_\ell
    M^\rT
        M
            \Theta_\ell
            =
            \|M\|_\rF^2 I_3 - M^\rT M
\end{equation}
is a real positive semi-definite symmetric matrix of order $3$. Its spectrum is given  by
\begin{equation}
\label{spec}
    \{\lambda_1+\lambda_2,\, \lambda_1+\lambda_3,\, \lambda_2+\lambda_3 \},
\end{equation}
where $\lambda_1, \lambda_2, \lambda_3\> 0$  are the eigenvalues of the $(3\x 3)$-matrix $M^\rT M \succcurlyeq 0$. Therefore, if $M$ satisfies the rank condition
\begin{equation}
\label{rankM}
    \mathrm{rank }M \> 2,
\end{equation}
then at most one of the eigenvalues of $M^\rT M$ is zero, and hence, the spectrum (\ref{spec}) is all strictly positive, thus implying that $\Lambda\succ 0$. The algebraic Lyapunov inequality
$$
    A+A^\rT = -4\Lambda \prec 0,
$$
obtained by combining (\ref{AA}) with (\ref{MM}), leads to the matrix $A$ being Hurwitz. This makes (\ref{rankM}) a sufficient condition for the asymptotic stability of the quantum plant with the Pauli matrices (\ref{XPauli}).

\section{Conclusion}\label{sec:conc}

We have considered a class of open quantum stochastic systems whose Hamiltonian and coupling operators are linear and affine functions of dynamic variables with algebraic properties, similar to and extending those of the Pauli matrices. The linearity of the drift vector and the dispersion matrix in the resulting quasilinear  QSDE gives rise to tractable dynamics of mean values
and higher-order multi-point moments of the system variables in the case of vacuum input fields. This also allows the invariant quantum state to be studied through the method of moments and makes quadratic and more general cost functionals with Lur'e type nonlinearities,  and their growth rates,  effectively computable for such systems. 
A mean square optimal measurement-based filtering problem for quasilinear quantum plants has been solved  in a class of Luenberger observers, leading to a Kalman-like quantum filter.  In regard to a steady-state filtering regime, a rank condition on the coupling matrix has been obtained for stability of quantum plants with the Pauli matrices as initial variables.  These results  can also be extended to quantum feedback control problems with applications to physical settings which involve interaction of particle spins with  electromagnetic fields.

\end{document}